\documentclass[12pt]{article}
\usepackage{amsmath}
\usepackage{graphicx,psfrag,epsf}
\usepackage{enumerate}
\usepackage[round,authoryear,longnamesfirst]{natbib}
\usepackage{url} 
\usepackage{amsfonts,amssymb,amsthm,bm,accents,statex,titling}
\usepackage{marginnote,rotating,fancyvrb}
\usepackage[en-US]{datetime2}
\usepackage{float,cases}
\usepackage{multirow,dcolumn,booktabs,caption,threeparttable,longtable,tabularx}
\usepackage{siunitx}
\usepackage{hyperref}
\usepackage{bbm}
\usepackage{subcaption}
\usepackage{mathrsfs}
\usepackage[linesnumbered,ruled]{algorithm2e}
\usepackage[svgnames]{xcolor}
\usepackage[doublespacing]{setspace}
\usepackage[includehead]{geometry}
\usepackage{xr}
\usepackage{standalone}

\newcommand{\Expt}{\mathop{\mbox{\sf E}}}
\newtheorem{theorem}{Theorem} 

\newtheorem{lemma}{Lemma}

\newtheorem{remark}{Remark}

\newtheorem*{assumption*}{\assumptionnumber}
\providecommand{\assumptionnumber}{}
\makeatletter

\def\widebar{\accentset{{\cc@style\underline{\mskip12mu}}}}
\def\@pdfborder{0 0 0}

\newenvironment{assumption}[1]
{
	\renewcommand{\assumptionnumber}{Assumption #1}%
	\begin{assumption*}
		\protected@edef\@currentlabel{#1}%
	}
	{
	\end{assumption*}
}
\hypersetup{colorlinks, citecolor=NavyBlue, filecolor=NavyBlue, linkcolor=NavyBlue, urlcolor=NavyBlue}

\newcommand{\blind}{0}

\graphicspath{{./Figs/}}

\begin{document}

\def\spacingset#1{\renewcommand{\baselinestretch}%
{#1}\small\normalsize} \spacingset{1}


\if0\blind
{
  \title{\bf A Time-Varying Network for Cryptocurrencies}
  \author{Li Guo\thanks{600 Guoquan Rd, Shanghai 200433. Email: guo\_li@fudan.edu.cn.}\\
    {\small Fudan University} \\
    {\small Shanghai Institute of International Finance and Economics} \\
    \\
	Wolfgang Karl H{\"a}rdle\thanks{Unter den Linden 6 10099 Berlin, Germany. Email: haerdle@hu-berlin.de.} \\
	{\small Humboldt-Universit{\"a}t zu Berlin, Singapore Management University} \\
	{\small Xiamen University, Charles University} \\
    \\
    Yubo Tao\thanks{Corresponding author: 90 Stamford Rd, Singapore 178903. Email: ybtao@smu.edu.sg.} \\
    {\small Singapore Management University}}
  \maketitle
}\fi

\if1\blind
{
  \vspace*{2cm}
  \begin{center}
    {\LARGE\bf A Time-varying Network for Cryptocurrencies}
  \end{center}
  \medskip
} \fi

\begin{abstract}
	Cryptocurrencies return cross-predictability and technological similarity yield information on risk propagation and market segmentation. To investigate these effects, we build a time-varying network for cryptocurrencies, based on the evolution of return cross-predictability and technological similarities. We develop a dynamic covariate-assisted spectral clustering method to consistently estimate the latent community structure of cryptocurrencies network that accounts for both sets of information. We demonstrate that investors can achieve better risk diversification by investing in cryptocurrencies from different communities. A cross-sectional portfolio that implements an inter-crypto momentum trading strategy earns a 1.08\% daily return. By dissecting the portfolio returns on behavioral factors, we confirm that our results are not driven by behavioral mechanisms.
\end{abstract}

\noindent%
\vfill

\newpage
\spacingset{1.4} 
\section{Introduction}
\label{sec:intro}

The invention of Bitcoin has spurred the creation of many cryptocurrencies (cryptos hereafter), commonly known as \textit{Altcoins}. According to the statistics available at \textit{coinlore.com}, as of July 2021, a total of 5,982 cryptos (coins and tokens) are actively traded worldwide, with a total market capitalization of over 1.2 trillion dollars. The growing number of cryptos provides investors with more alternatives to form crypto portfolios that deliver potentially higher Sharpe ratios \citep{Markowitz1952, Roy1952, Sharpe1966}. Therefore, it is essential to understand the inter-relationships between the cryptos for achieving optimal investment decisions.

For traditional assets, a natural starting point is to compare their fundamental similarities. For example, we can rely on industrial classifications based on business operations, e.g., Standard Industrial Classification (SIC) and the Global Industry Classification Standard (GICS), or product differentiation \citep{Hoberg2016} to infer the similarity between stocks. Unlike stocks and many other traditional assets, the fundamentals of cryptos are difficult to determine as they neither deliver dividends nor release accounting information. In addition, many investors invest in cryptocurrencies because they believe in the blockchain technology embodied in these digital coins \citep{LiuShengWang2021}. However, because of the open-source nature of blockchain technology, technological similarities could be uninformative since many \textit{Clonecoins} are using the same technology but have utterly different price trajectories. 

To resolve the dilemma, we introduce a novel method to classify cryptos by adding the return cross-predictability to technological similarities. The benefits are twofold. First, return cross-predictability provides timely information to understand the dynamics of market structure and the information propagation across financial assets \citep{Menzly2010}. Second, the returns embody investors' beliefs, which is crucial in classifying cryptos or other hard-to-value financial assets \citep{Detzel2021}. We naturally employ a directed time-varying network to model the evolution of return cross-predictability and contextualize the cryptos with the technological characteristics. Using network analysis, we develop a dynamic covariate-assisted spectral clustering (CASC) method, which accommodates important network features, such as degree heterogeneity, directionality, and time-varying membership, to systematically study the inter-relationships between cryptos. We also provide a theory and conduct extensive simulations to illustrate the classification consistency of our method.

We carry out three empirical tests to illustrate the economic contribution of the crypto communities. First, we demonstrate that the investors can achieve better risk diversification by investing in cross-community cryptos. Classical portfolio theory \citep{Markowitz1952, Roy1952} suggests that investors can achieve the expected portfolio return with a lower risk (variance) by investing in less positively correlated assets. Therefore, we calculate and compare the within- and cross-community average pairwise return correlations using 7-day and 30-day future crypto returns. It shows that the return correlations between cryptos within the same community are, on average, significantly higher than those across communities. The pattern is robust to different evaluation time horizons and only valid for dynamic CASC-estimated communities.   

Second, to formally test the information propagation effects within crypto communities, we form a cross-sectional portfolio that implements an inter-crypto momentum trading strategy. Prior literature has documented that, due to information propagation, assets that have fundamental similarities will have momentum spillovers, wherein past return of one asset predicts the returns of assets linked to it \citep[see, e.g.,][]{Moskowitz1999, Cohen2008, Menzly2010, Lee2019, Ali2020, Parsons2020}. Based on this argument, we construct a trading signal for each crypto using the average returns of the rest of cryptos in the same community. We sort the cryptos into quartiles according to the trading signals and implement a momentum trading strategy through simultaneously buying the cryptos in the top quartile and selling the cryptos in the bottom quartile. It shows that the trading strategy generates a one-day-ahead average daily return of 1.08\%, and the return does not revert in a week. 

Finally, we explore the behavioral mechanisms, such as limit-to-arbitrage, investor attention, and macro uncertainty, as alternative explanations to momentum spillover effects. Specifically, we split the whole sample into high and low episodes of the behavioral factors by cutting off at their medians. It shows that the momentum portfolio returns exist regardless of these factors, indicating that behavioral mechanisms could not be the drivers of momentum spillover effects. 

This paper makes several fundamental contributions to the empirical understanding of the cryptos market. To the best of our knowledge, our paper is the first to establish a crypto network and study the cryptos market segmentation problem using community detection. In addition, we are among the first to empirically justify the importance of crypto technologies in determining crypto prices. By creatively combining the return cross-predictability and technological similarities, we identify the community pattern between cryptos and document an inter-crypto momentum spillover effect, which is direct evidence for network effects studied in \cite{SockinXiong2020} and \cite{Cong2021b}. Our paper also contributes to community detection methods in network analysis by extending the spectral clustering method to identify communities in a time-varying network in the presence of both time-varying memberships and node covariates.

This paper is closely related to a rapidly growing literature on the empirical asset pricing of cryptocurrencies. \cite{Yermack2017} is the first paper to study the economic implications of blockchain in the context of corporate governance. \cite{Makarov2020} investigated cross-exchange crypto arbitrage and document a sizable price difference between exchanges. \cite{Griffin2020} analyzed the wallet-level blockchain data and examine the price manipulation in Bitcoin. \cite{LiuTsyvinski2021} provided a comprehensive analysis of the risk-return tradeoff of cryptos. \cite{LiuTsyvinskiWu2021} examined the cross-section of crypto and construct a three-factor model. \cite{LiuMarsh2021} explored the factor structure in both crypto returns and volatility and establish nine stylized facts in the cryptos market. Our paper provides new insights into the fundamentals of the cryptos market structure by dividing cryptos into different communities and document an inter-crypto momentum spillover effect. 

Our paper contributes to the literature on community detection. Namely, community detection is one of the central problems in network analysis, and various algorithms have been proposed to solve this problem. For global approaches, there are spectral clustering \citep{Rohe2011, Lei2015, Jin2015, JosephYu2016, Rohe2016, GaoMaZhangZhou2018, Pensky2019, ZhouAmini2019} and convex relaxations via semidefinite programs (SDPs) \citep{CaiLi2015, ChenXu2016, Hajek2016a, Hajek2016b, AminiLevina2018, Li2021}. For local approaches, there are acyclic belief propagation (ABP) \citep{Decelle2011}, Bayesian MCMC and variational Bayes \citep{Nowicki2001, Celisse2012, Bickel2013, Matias2017, WangBickel2017}, profile likelihood \citep{BickelChen2009, ZhaoLevinaZhu2012}, and pseudo-likelihood maximization \citep{Amini2013, ZhouAmini2020}. Our method modifies the spectral clustering algorithm to incorporate historical network structures, time-varying membership, and node covariates. The simulation confirms its superior classification accuracy over state-of-the-art methods.

In addition, our paper echos the literature that utilizes auxiliary information for network analysis (contextual network analysis). One strand of literature has tried to understand the network dependencies with the assistance of covariates \citep{Zou2017, Lan2018, Zou2021, Zhao2021, ZhuCaiMa2021}. Another strand of literature has studied the community detection under the contextual setting \citep{Zhang2016, WengFeng2016, Binkiewicz2017, MaMa2017, ZhangRohe2018, Deshpande2018, Mu2020, Abbe2020, YanSarkar2020, Esmaeili2021}. Our paper is closely related to \cite{Binkiewicz2017} and \cite{ZhangRohe2018}, where we extend the covariate-assisted spectral clustering method to estimate a time-varying community structure. 

The remainder of the paper is organized as follows. In Section \ref{SecConNet}, we construct a time-varying network contextualized by crypto fundamentals to model the information propagation in the cryptos market. In Section \ref{SecMethod}, we develop the dynamic covariate-assisted spectral clustering algorithm to estimate the crypto communities and demonstrate its effectiveness via both theory and simulations. In Section \ref{SecEcoMean}, we explore the economic implications of the crypto communities. We present our conclusions in Section \ref{SecConcl}. All proofs and technical details are provided in the supplementary appendix. The R codes to implement the algorithms are available at QuantNet (search keywords ``CASC").

\section{Construction of the Cryptocurrencies Network} \label{SecConNet}

\subsection{Network of information propagation}

Traditional asset pricing theory states that the cross-predictability of asset returns reflects the correlations in the fundamental determinants caused by common information of the future cash flows or discount rates \citep{Menzly2010}. Extensive literature has applied network analysis to the equity market to gain insights about information propagation and spillover of behavioral bias across the firms \citep[see, e.g.,][]{Cohen2008, Aobdia2014, Acemoglu2015, Herskovic2018, ChenHardleOkhrin2019}. Most recently, \cite{SockinXiong2020}, \cite{Cong2021a}, and \cite{Cong2021b} develop models for platform tokens and study their network effects caused by decentralization. Based on the above literature, we find it essential to adopt a network view on the crypto market to study how price information or shock transfers from one crypto to another. 

Following \cite{Menzly2010}, we model the information propagation in the cryptos network using the return cross-predictability. We collect the historical daily prices, trading volumes, and contract attributes (e.g., algorithm, proof type, age, total coins, etc.) of the top 400 cryptos in market capitalization as of December 31, 2018 from an interactive platform (\textit{Cryptocompare.com}) with free API access. The sampling period is from January 1, 2016 to December 31, 2018. After excluding cryptos under 1 year old or with incomplete price and contract information, we obtain a sample of 182 cryptos. 

To form the linkages between cryptos, we regress a crypto's current return on the other cryptos' lagged returns in a 360-day (1-year) rolling window. We employ the adaptive Lasso \citep{Zou2006} to select the significant regressors and to estimate their coefficients simultaneously:
\begin{equation} \label{eqAL}
	\hat b^{*}_i = \arg\min \left\lbrace \left\lVert r^{s}_{i,t+1} - \alpha_i - \sum_{j \neq i} b_{i,j} {r_{j,t}^{s}} \right\rVert^2 + \lambda_i \sum_{j \neq i} \hat {w}_{i,j} |{b_{i,j}}| \right\rbrace,
\end{equation}
where $r^{s}_{j,t}$ is the standardized return for crypto $j$, $\hat{b}^*_i = (\hat{b}^*_{i,1}, \cdots, \hat{b}^*_{i, N})^{\top}$ is the adaptive Lasso estimate, $\lambda_i$ are non-negative regularization parameters, and $\hat{w}_{i,j}$ are the weights corresponding to $|b_{i,j}|$ for $j = 1, \cdots, N$ in the penalty term. Conventionally, one defines $\hat{w}_{i,j} = 1/|\hat{b}^{\textit{LS}}_{i,j}|^{\gamma}$ with some $\gamma > 0$, where $\hat{b}^{\textit{LS}}_{i,j}$ denotes the coefficient estimated by ordinary least squares. The Lasso technique yields an active set that has ``parental'' influence on the focal crypto. Thus, we obtain an adjacency matrix for each period, $A_t$, $t = 1, \cdots, T$. Figure \ref{Fig-RetAdj} presents a sub-network of 20 cryptos with linkages estimated by regression (\ref{eqAL}) on selected dates in January 2018. 

\begin{figure}[htp!]
	\centering
	\begin{subfigure}[tph!]{0.48\textwidth}
		\includegraphics[width = \linewidth]{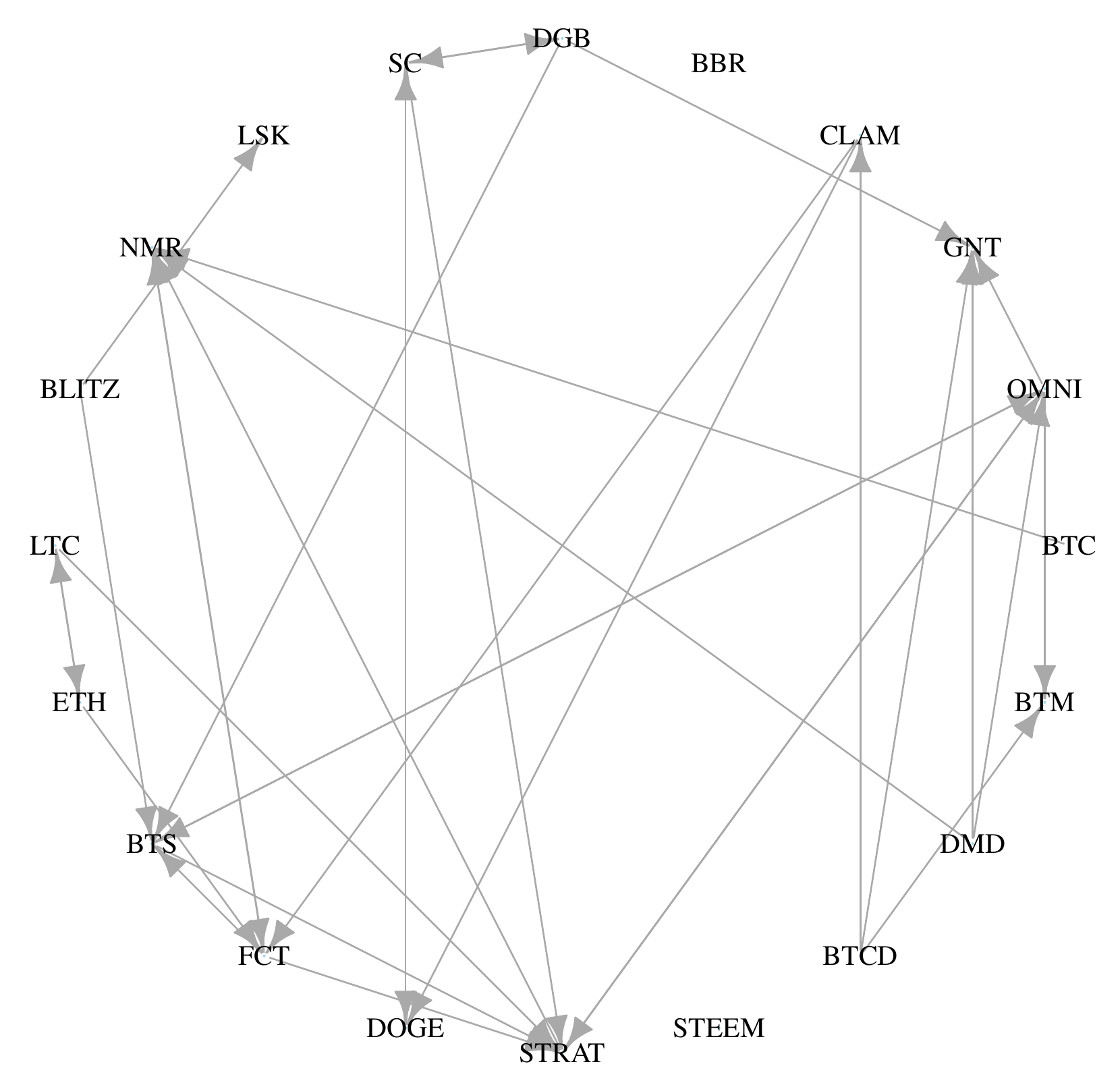}
		\caption{2018-01-01}
	\end{subfigure}
	\begin{subfigure}[tph!]{0.48\textwidth}
		\includegraphics[width = \linewidth]{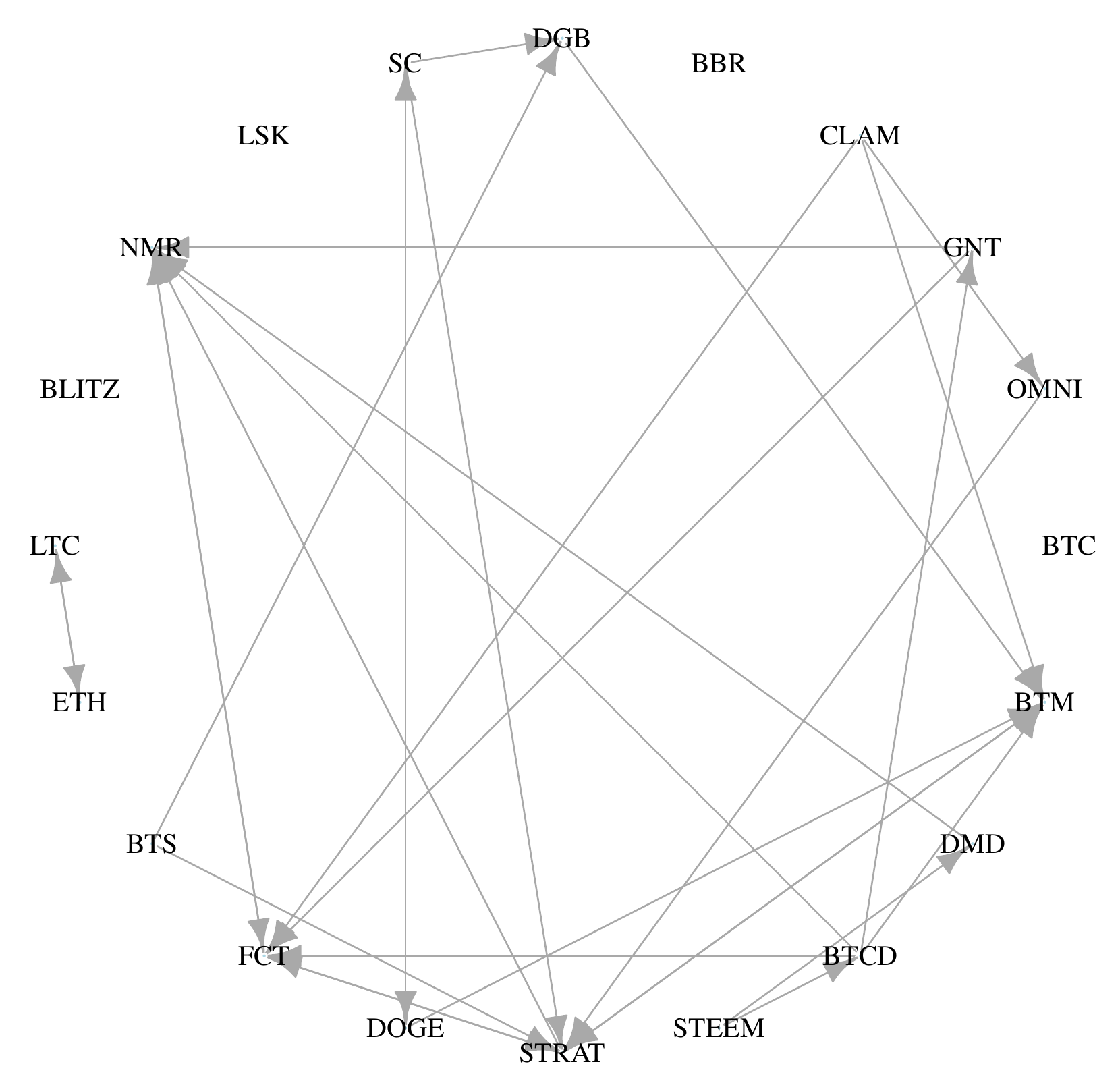}
		\caption{2018-01-05}
	\end{subfigure}\\
	
	\begin{subfigure}[tph!]{0.48\textwidth}
		\includegraphics[width = \linewidth]{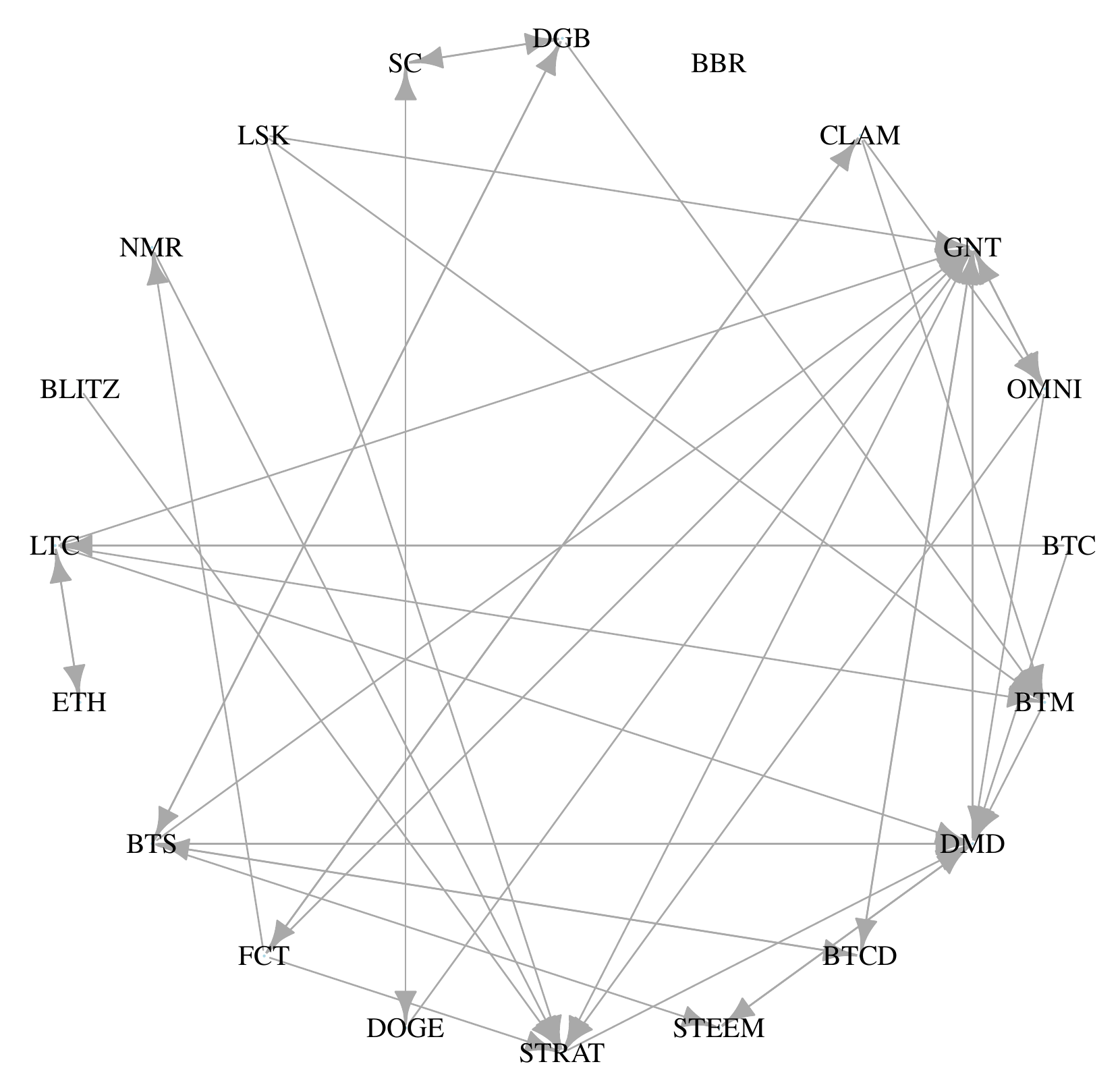}
		\caption{2018-01-15}
	\end{subfigure}
	\begin{subfigure}[tph!]{0.48\textwidth}
		\includegraphics[width = \linewidth]{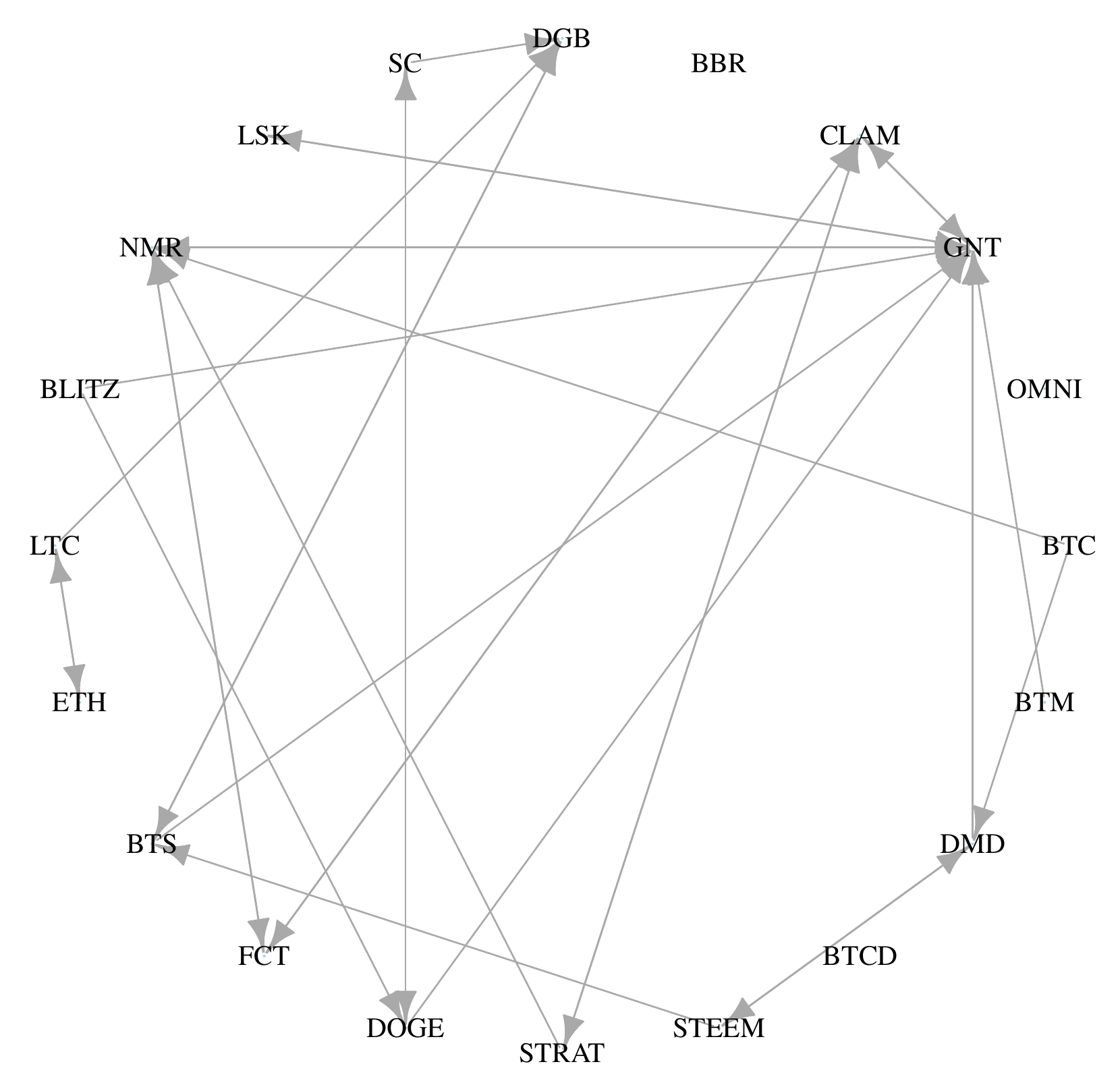}
		\caption{2018-01-31}
	\end{subfigure}
	\caption{Networks Formed by Return Cross-predictability. This figure presents the return-based cryptos network on selected dates in January 2018. We selected 20 cryptos, including BTC, ETH, LTC, and other top cryptos by market capitalization as of December 31, 2017. We obtained the connections from the predictive regression $r^{s}_{i, t+1} = \alpha_i + \sum_{j = 1, j \neq i}^{N - 1} b_{i,j} r^{s}_{j,t} + \epsilon_{i, t},$ where $r^{s}_{i,t}$ is the standardized daily return on crypto $i$, and $N$ is the total number of cryptos. Only the cryptos selected by adaptive Lasso will be linked to crypto $i$. An arrow from crypto A to crypto B means crypto A's return can predict the crypto B's return. }
	\label{Fig-RetAdj}
\end{figure}

Three critical features can be observed from the return-inferred network structure. First, the return cross-predictability is changing over time. For instance, the linkage between BTC and NMR on January 1 vanishes on January 5, and the predictive relationship between BTS and DGB on January 1 reverts on January 5. These observations indicate that information propagation is not always mono-directional or continues to exist between cryptos.

Second, the return cross-predictability induces a sparse network structure, that is, many isolated cryptos with no linkages in the network. The examples include BBR and STEEM on January 1; BBR, BLITS, BTC, and LSK on January 5; BBR on January 15; BBR, OMNI, and BTCD on January 31. The sparsity is an important feature that may create many challenges for conducting an in-depth analysis of the network structure, such as community detection, network formation modeling, network differential tests, and so forth. To resolve this issue, we need to explore new information sources to identify additional linkages between cryptos or apply advanced statistical methods designed for analyzing sparse networks. In fact, we use both strategies in the later analyses.

Finally, a larger market capitalization does not imply a leading position in return cross-predictability. For example, the largest three cryptos (BTC, LTC, and ETH) predict fewer crypto returns than CLAM on January 5 and LSK on January 15. These observations suggest that return cross-predictability is a noisy estimator for uncovering the authentic connections between cryptos, which motivates us to seek additional information sources to form crypto linkages so that the noises could cancel out with each other when we conduct a combinatorial analysis for the cryptos network.

\subsection{Contextualization with technology}

As pointed out by \cite{LiuShengWang2021}, although cryptos do not distribute dividends or release accounting information, many investors invest in them simply because they believe blockchain technology is an important innovation, and cryptos are assets that represent a stake in the future of this technology. Specifically, there are two building blocks for blockchain technology: hashing algorithms and consensus mechanism (or proof types).

\textit{Algorithm}, which is short for the \textit{hashing algorithm}, plays a central role in determining the crypto's security. For each crypto, there is a hash function in mining; for example, Bitcoin (BTC) uses double SHA-256 and Litecoin (LTC) uses Scrypt. As security is one of the most important features of cryptos, the hashing algorithm naturally--in terms of trust--determines the intrinsic value of a crypto. In the example above, the Scrypt system was used with cryptos to improve upon the SHA256 protocol. The SHA256 preceded the Scrypt system and was the basis for BTC. Specifically, Scrypt was employed as a solution to prevent specialized hardware from brute-force efforts to outmine others. Thus, Scrypt-based Altcoins require on average more computing effort per unit than the equivalent coin using SHA256. The relative difficulty of the algorithm confers a relative value.

\textit{Proof Types}, or proof system/protocol, is an economic measure to deter denial of service attacks and other service abuses such as spam on a network by requiring some work from the service requester, usually the equivalent to processing time by a computer. For each crypto, at least one of the protocols will be chosen as a transaction verification method; for example, BTC and Ethereum (ETH) currently use the Proof-of-Work (PoW), whereas Diamond (DMD) and Blackcoin use the Proof-of-Stake (PoS). PoW-based cryptos such as BTC use mining---the solving of computationally intensive puzzles---to validate transactions and create new blocks. In PoS-based cryptos, the creator of the next block is chosen through various combinations of random selection and wealth (in terms of crypto) or age (i.e., the stake). In summary, the proof protocol determines the reliability, security, and effectiveness of the transactions.

Recent literature has tried to model the mining process and different proof types under a theoretical framework. Many papers \citep[see, e.g.,][]{Biais2019, Easley2019, CongHeLi2021} have studied mining activities and transaction costs and their implications for crypto prices. Among them, \cite{Pagnotta2020} is the first to build an equilibrium model embedding a security channel that leads to multiple equilibria. \cite{Abadi2018}, \cite{Budish2018}, and \cite{Hinzen2019} studied the limitations and the pricing implications of the PoW mechanism, while \cite{Fanti2019} and \cite{Saleh2021} explored the PoS mechanism. Based on these models, we conjecture that the technological similarities could be vital for understanding the inter-relationships between cryptos. 

We collect the algorithm and proof type information for each crypto in our sample and contextualize the network induced by return cross-predictability with the technological similarities between cryptos. For illustration, we plot the sub-network in Figure \ref{Fig-TechAdj} using the same cryptos as in Figure \ref{Fig-RetAdj}, where the linkages are formed if the cryptos share at least one fundamental characteristic. For example, LTC and DOGE are linked because they both adopt the Scrypt algorithm. Evidently, due to the limited choices of algorithms and proof types, the cryptos are more likely to connect with each other when using technological similarities to form linkages. Intuitively, the network will be denser if we na{\"i}vely link cryptos either by return cross-predictability or technological similarity, such as adding linkages in subfigure (c) of Figure \ref{Fig-TechAdj} to each subfigure of Figure \ref{Fig-RetAdj}. In the following section, we will discuss how to combine the two sets of information in a more sophisticated manner. 

\begin{figure}[htp!]
	\centering
	\begin{subfigure}[tph!]{0.48\textwidth}
		\includegraphics[width = \linewidth]{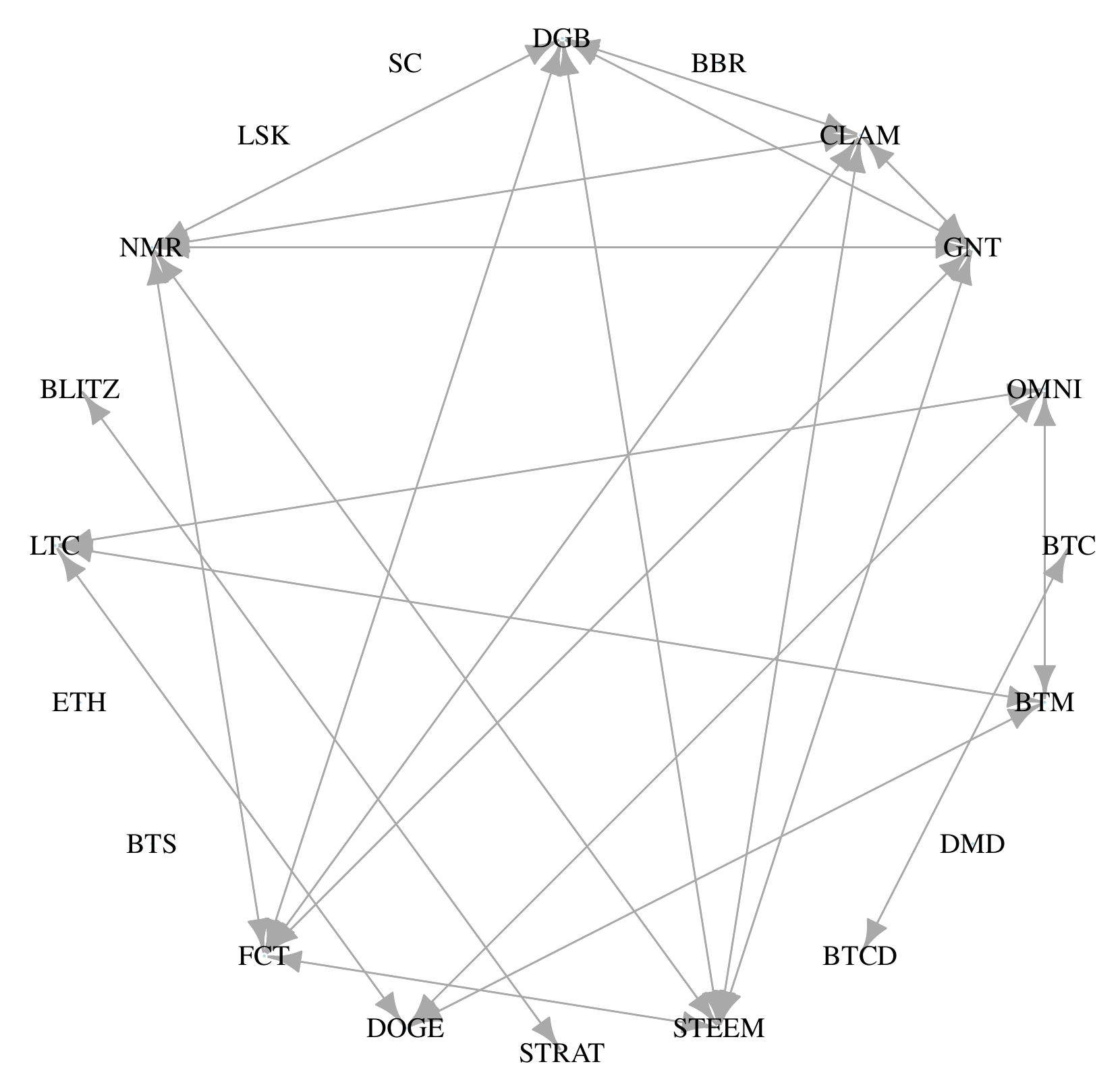}
		\caption{Algorithm}
	\end{subfigure}
	\begin{subfigure}[tph!]{0.48\textwidth}
		\includegraphics[width = \linewidth]{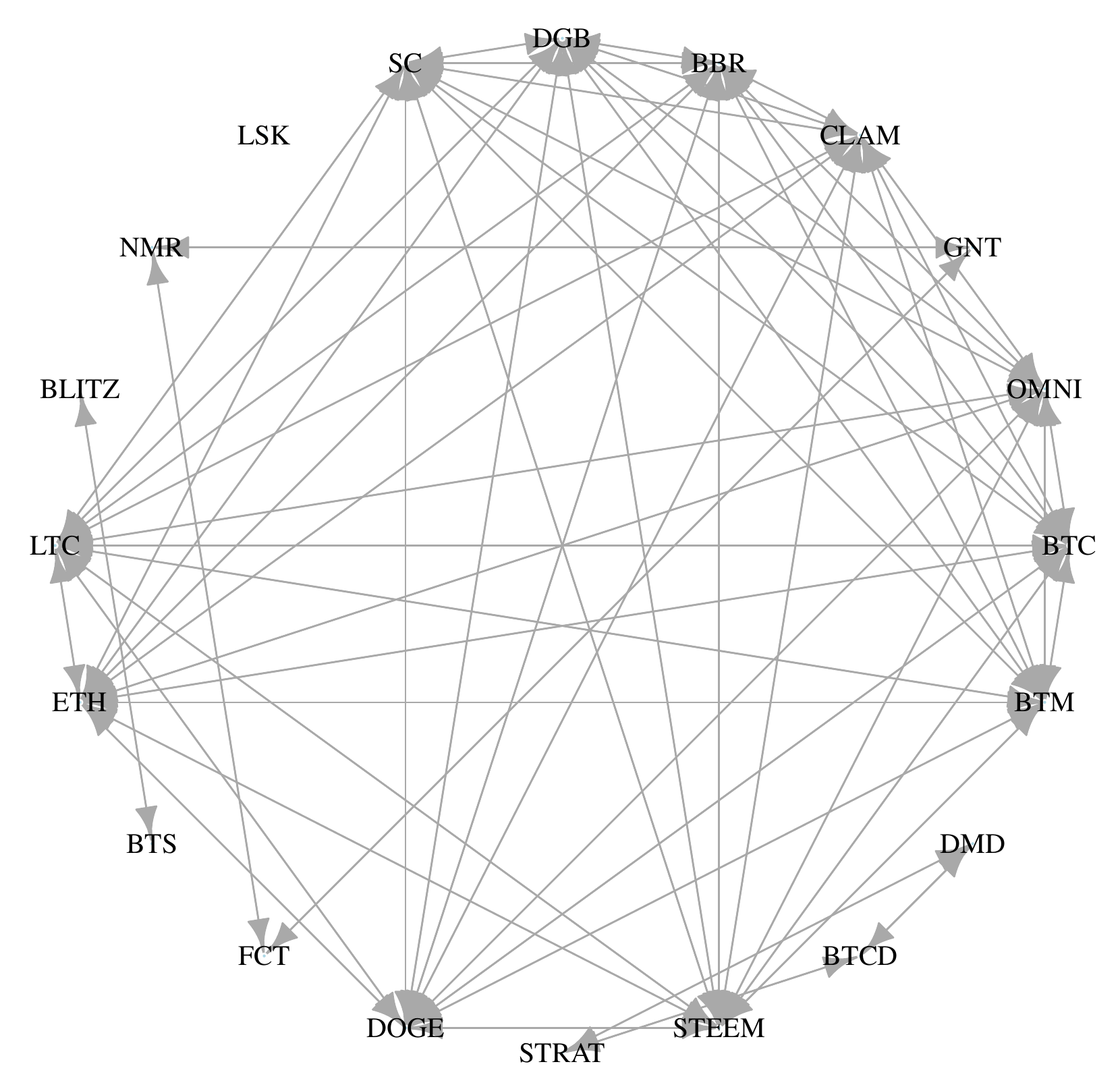}
		\caption{Proof Types}
	\end{subfigure}\\
	
	\begin{subfigure}[tph!]{0.48\textwidth}
		\includegraphics[width = \linewidth]{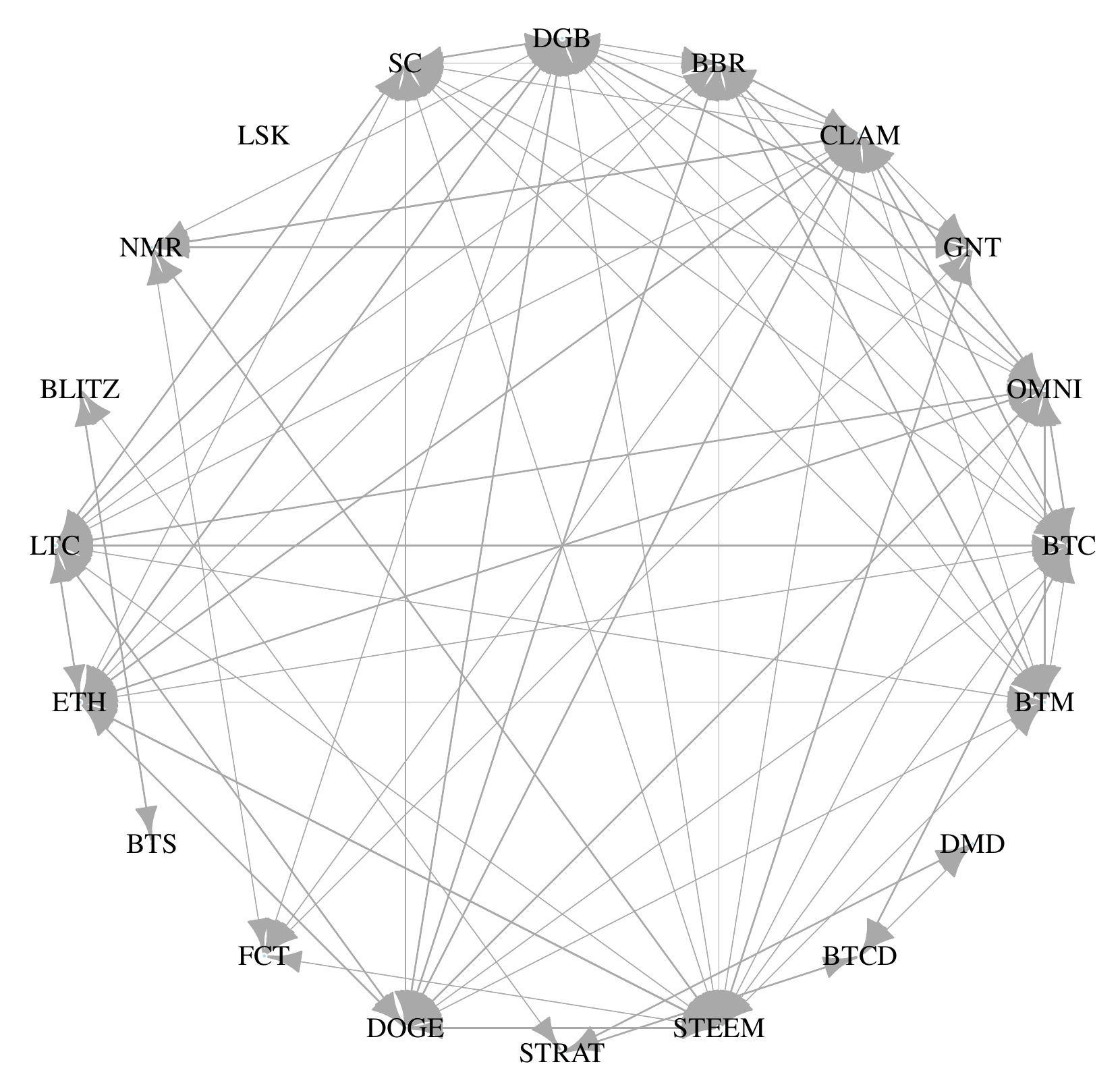}
		\caption{Combined Fundamentals}
	\end{subfigure}
	\caption{Networks Formed by Technological Similarities. This figure depicts the contract-based cryptos network. We connect two cryptos with double-sided arrows if they have the same algorithm or proof type.}
	\label{Fig-TechAdj}
\end{figure}

In summary, we have built a time-varying cryptos network that characterizes both return cross-predictability and technological similarity between cryptos. The price shocks (e.g., policy changes, macro events, etc.) will travel through return-based linkages, while technological shocks (e.g., innovation, security breaches, etc.) will travel through technology-based connections. In the following sections, we will investigate deeper on the inter-relationships between cryptos based on this network.

\section{Communities in the Cryptocurrencies Network} \label{SecMethod}

Classical empirical studies on asset pricing have pointed out the importance of market segmentation in generating heterogeneous prices of risk and profitable trading strategies. However, most of them use \textit{ex-ante} information, such as industry, listing exchange, and geography, to conjecture market segments \citep[see, e.g.,][]{FamaFrench1993, FamaFrenchBoothRex1993, Foerster1999, Griffin2002}. Recently, \cite{PattonWeller2019} proposed a data-driven method (EM algorithm) to cluster stocks with similar risk prices and provide strong evidence of market segmentation across all portfolios, factor pricing models, and time periods. Based on these findings, we deduct market segmentation also plays an important role in understanding the risks and returns of cryptos and would like to cluster cryptos under the network framework.    

A standard framework for studying clustering effect, that is, community detection, in a network setting is the stochastic blockmodel (SBM) proposed by \cite{Holland1983}. Using this framework, we can represent the adjacency matrices stochastically via a block structure and adopt spectral clustering to identify latent communities. \cite{Binkiewicz2017} and \cite{ZhangRohe2018} showed that, by introducing covariate assistance, they could improve the classification accuracy of the spectral clustering method under static network settings. This section presents an extension of the covariate-assisted spectral clustering (CASC) algorithm to deal with the time-varying directed networks. Theoretical justifications and simulations demonstrate the consistency of this method.

\subsection{Dynamic contextual stochastic blockmodel} \label{SubsecModel}

As return cross-predictability between any two cryptos could be mono-directional, we consider a time-varying network defined as a sequence of random directed graphs with $N$ nodes, $G_{N,t}$, $t=1,\cdots,T$, on the vertex set $V_N = \{v_1, v_2, \cdots, v_N\}$, which does not change over time. For each period, we model the directed network structure with an adjacency matrix $A_t$ which is mostly asymmetric, that is,
\begin{equation}
	{A}_t(i,j) = \begin{cases}
		\text{Bernoulli}\{P_{t}(i,j)\}, &\text{ if $i \neq j$} \\
		0, &\text{ if $i = j$}
	\end{cases},
\end{equation}
where ${P}_{t}(i,j) = \Pr\{A_t(i,j) = 1\}$. The adjacency matrix $A_t$ records the pattern of edges that are induced by crypto return cross-predictability in network at time $t$; if there is an edge starting from node $i \in \{1, \cdots, N\}$ and ending at node $j \in \{1, \cdots, N\}$ (i.e., $i \rightarrow j$), then $A_t(i,j) = 1$; otherwise, $A_t(i,j) = 0$. Thus, the $i$th row of $A_t$ records how crypto $i$ predicts other cryptos (sending edges), and the $j$th column of $A_t$ records how crypto $j$ is predicted by other cryptos (receiving edges). 

To model the community structure, the probabilities of a connection ${P}_t(i,j)$ at period $t$ are partitioned into blocks. Let the block probability matrix in each period be $B_t \in [0,1]^{K_R \times K_C}$ with rank $K = \min\{K_R, K_C\}$. In particular, let us denote $z_{R,it}$ as the row community label of node $i$ at time $t$ and $z_{C,it}$ as the column community label of node $j$ at time $t$, respectively; then, if $z_{R,it} = k_R$ and $z_{C,jt} = k_C$, then ${P}_{t}(i,j) = {B}_t(z_{R,it}, z_{C,jt}) = {B}_{t}(k_R,k_C)$. Hence, for any $t = 1, \cdots, T$, we can obtain the population adjacency matrix 
\begin{equation}
	\mathcal{A}_t \stackrel{\operatorname{def}}{=} \Expt(A_t) = {Z}_{R,t}{B}_t{Z}_{C,t}^\top,
\end{equation}
where $Z_{R,t} \in \{0,1\}^{N \times K_R}$ and $Z_{C,t} \in \{0,1\}^{N \times K_C}$ are the \textit{row-clustering matrices} and \textit{column-clustering matrices}, respectively. Notably,  there is only one 1 in each row and at least one 1 in each column of each clustering matrix. The co-clustering of the adjacency matrix yields two partitions of the same set of nodes. The row communities contain nodes with similar sending patterns, and the column communities contain nodes with similar receiving patterns. 

\cite{QinRohe2013} and \cite{JosephYu2016} showed that the regularization of adjacency matrix can substantially improve the spectral clustering performance, especially for sparse networks. Therefore, we define the regularized graph Laplacian $L_{\tau,t} \in \mathbbm{R}^{N \times N}$ for the directed network as
\begin{equation}
	L_{\tau,t} = D_{R,t}^{-1/2}A_tD_{C,t}^{-1/2},
\end{equation}
where $D_{R,t}$ and $D_{C,t}$ are diagonal matrices with $D_{R,t}(i,i) = \sum_{j=1}^{N} A_{t}(i,j) + \tau_{R,t}$ and $D_{C,t}(i,i) = \sum_{j=1}^{N}A_{t}(j,i) + \tau_{C,t}$, where $\tau_{R,t}$ and $\tau_{C,t}$ are set to be the average row and column degrees at each period, respectively.

To contextualize the network with the technological similarities between cryptos, we introduce the covariate matrix $X \in \{0,1\}^{N \times R}$, where each row of $X$ denotes a crypto, and each column denotes a type of technology. Thus, $X(i,j)$ equals 1 if crypto $i$ adopts technology $j$ and is 0 otherwise. Following \cite{Binkiewicz2017}, we assume the covariate $X$ also reflects the row- and column-cluster's information, that is, let $M_{R,t} \in [0,1]^{K_R \times R}$ and $M_{C,t} \in [0,1]^{K_C \times R}$ denote the covariate expectation matrix for row- and column-clusters in each period $t$, respectively, then the population covariate matrix is
\begin{equation} \label{eqCov}
    \mathcal{X} = \Expt(X) = Z_{R,t}M_{R,t} = Z_{C,t}M_{C,t}.
\end{equation}

Then, we combine return cross-predictability and technological similarities between cryptos by constructing a similarity matrix from regularized graph Laplacian $L_{\tau,t}$ and covariate matrix $X$, that is, for each $t = 1, \cdots, T$,
\begin{equation} \label{eqSim}
	S_t = L_{\tau,t} + \alpha_t C^w_t = D_{R,t}^{-1/2}A_tD_{C,t}^{-1/2} + \alpha_t XW_tX^{\top},
\end{equation}
where $\alpha_t \in [0, \infty)$ is the tuning parameter that controls for the informational balance between return cross-predictability and technological similarities, and $W_t$ is a time-varying weight matrix that controls for the degree of technological similarities.

The setup in Equation (\ref{eqSim}) addresses several extensions of existing methods. First, $W_t$ creates a time-varying interaction between different covariates. For instance, we may think of different refined algorithms that have the same origin. Such inheritance relationships will potentially lead to an interaction between the cryptos. In addition, as shown in Figure \ref{Fig-TechEv}, some algorithms may become more popular over time, while the others may be near extinction. Thus, this interaction would also change over time. Second, we can easily select covariates by setting certain elements of $W_t$ to 0. This is necessary as it helps us to model the evolution of technologies. At some point in time, some technology may disappear due to upgrades or cracking. Therefore, $W_t$ offers us the flexibility to exclude covariates. Finally, due to the open-source nature of the blockchain, crypto developers can easily copy and paste the source code and launch a new coin without any costs. Consequently, we observe a high degree of homogeneity in the cryptos market. However, this homogeneity does not necessarily result in a cross-predictability of prices: some cryptos are negatively correlated. In this case, we may set $W_t(i,i)$ to be negative, and $C^w_t$ will eventually bring the cryptos with different technologies closer in the similarity matrix.

\begin{figure}[htp!]
    \centering
    \includegraphics[width = \textwidth]{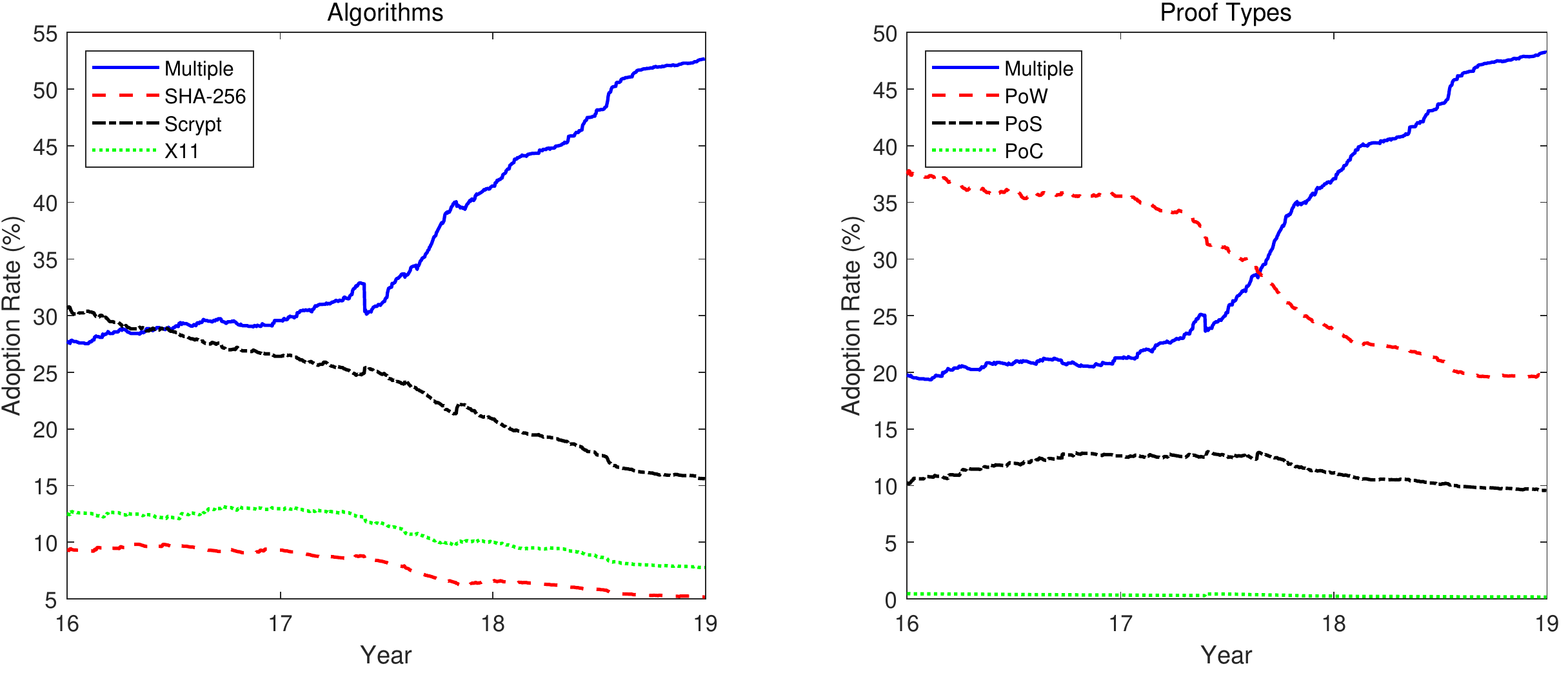}
    \caption{\small The Evolution of Technology Adoption Rates. This figure presents the adoption rates (or the popularity) of the major technologies in cryptos market from January 1, 2016 to December 31, 2018. }
    \label{Fig-TechEv}
\end{figure}

Following \cite{Karrer2011}, we employ degree correction to allow for heterogeneous degrees. Specifically, we introduce two sets of degree parameters $\psi^R = (\psi^R_1,\psi^R_2,\cdots,\psi^R_N)$ and $\psi^{C} = (\psi^C_1,\psi^C_2,\cdots,\psi^C_N)$ so that the edge probability are given by
\begin{equation*}
	P_t(i,j) = \psi^{R}_i\psi^{C}_jB_t(i,j)
\end{equation*}
under the restriction that
\begin{equation*}
	\sum_{i \in \mathcal{G}^R_{t,k}} \psi^R_{i} = 1, \forall k \in \{1, 2, \cdots, K_R\}, \quad	\sum_{i \in \mathcal{G}^C_{t,k}} \psi^C_{i} = 1, \forall k \in \{1, 2, \cdots, K_C\}.
\end{equation*}
Then, the population adjacency matrix for dynamic degree-corrected contextual stochastic blockmodel (DCCBM) is
\begin{equation}
	\mathcal{A}_t = \Expt(A_t) = \varPsi^{R}Z_{R,t}B_tZ_{C,t}^{\top}\varPsi^{C},
\end{equation}
where $\varPsi^{s} = \text{Diag}(\psi^s), s \in \{R, C\}$, and the population regularized graph Laplacian is
\begin{equation}
	\mathcal{L}_{\tau, t} = \mathcal{D}_{R,t}^{-1/2}\mathcal{A}_t\mathcal{D}_{C,t}^{-1/2}.
\end{equation}
Therefore, the population similarity matrix is 
\begin{equation} 
	\mathcal{S}_t = \mathcal{L}_{\tau, t} + \alpha_t \mathcal{C}^w_t,
\end{equation}
where $\mathcal{C}^w_t = \mathcal{X}\mathcal{W}_t\mathcal{X}^{\top}$.

By construction, we know that $D_{R,t}(i,i) = \sum_{j=1}^{N} \mathbf{1}^{(t)}_{\{i\rightarrow j\}} + \tau_{R,t}$, which controls for the number of the parents of node $j$, and that $D_{C,t}(j,j) = \sum_{i=1}^{N} \mathbf{1}^{(t)}_{\{j \rightarrow i\}} + \tau_{C,t}$, which controls the number of the offspring of node $j$. To analyze the asymmetric adjacency matrix $A_t$ caused by directional information, \cite{Rohe2016} proposed using the \textit{singular value decomposition} instead of eigen-decomposition for the regularized graph Laplacian. The intuition behind this methodology is to use both the eigenvectors of $L_{\tau,t}^{\top}L_{\tau,t}$ and $L_{\tau,t}L_{\tau, t}^{\top}$, which contains information about ``the number of common parents'' and ``the number of common offsprings''; that is, for each $t = 1, \cdots, T$:
\begin{align*}
	(L_{\tau,t}^{\top}L_{\tau,t})_{ab} &= \sum_{i=1}^{N}L_{\tau,t}(i,a)L_{\tau,t}(i,b) = \dfrac{1}{\sqrt{D_{C,t}(a,a)D_{C,t}(b,b)}}\sum_{i=1}^{N}\dfrac{\mathbf{1}^{(t)}_{\{i \rightarrow a \and i \rightarrow b\}}}{D_{R,t}(i,i)},\\
	(L_{\tau,t}L_{\tau,t}^{\top})_{ab} &= \sum_{i=1}^{N}L_{\tau,t}(a,i)L_{\tau,t}(b,i) = \dfrac{1}{\sqrt{D_{R,t}(a,a)D_{R,t}(b,b)}} \sum_{i=1}^{N}\dfrac{\mathbf{1}^{(t)}_{\{a \rightarrow i \and b \rightarrow i\}}}{D_{C,t}(i,i)}.
\end{align*}
Therefore, the column communities consist of cryptos that receive similar price and technology shocks (common parents), and the row communities contain cryptos transmitting shocks to similar cryptos (common offspring). 

\subsection{Dynamic covariate-assisted spectral clustering} \label{SubsecMethod}

When setting up a dynamic covariate-assisted spectral clustering algorithm, we face two major difficulties: (i) determining the interactive weights $W_t$ between covariates and (ii) estimating the similarity matrix with time-varying network information. For the first issue, we use the adoption rates to present the popularity of the crypto technology. Specifically, for any two technologies (algorithms or proof types) $a$ and $b$ at time $t$, we have
\begin{equation}
    W_t(a,b) = \dfrac{N_{a,t}}{N_t} \times \dfrac{N_{b,t}}{N_t},
\end{equation}
where $N$ is the total number of cryptos in the market, $N_a$ and $N_b$ are the numbers of cryptos that adopt technology $a$ and $b$, respectively. Intuitively, the weight is greater when the technologies are more popular among the cryptos, indicating a higher strength of technological similarity. 

For the second issue, we follow \cite{Pensky2019} by constructing the estimator of $\mathcal{S}_t$ with a discrete kernel to bring in historical linkage information. Specifically, we first pick an integer $r \geq 0$, obtain a set of integers $\mathcal{F}_{r} = \{-r,\cdots, 0\}$ and assume that $\big\Vert W_{r,\ell}(i) \big\Vert \leq W_{\max}$, where $W_{\max}$ is independent of $r$ and $i$, and satisfies
\begin{equation}
	\dfrac{1}{|\mathcal{F}_{r}|}\sum_{i \in \mathcal{F}_{r}}i^k W_{r,\ell}(i) = \begin{cases}
		1, & \text{if $k = 0$}, \\
		0, & \text{if $k = 1, 2, \cdots, \ell-1$}. \\
	\end{cases}
\end{equation}

\begin{figure}[htp!]
    \centering
    \includegraphics[width = \textwidth]{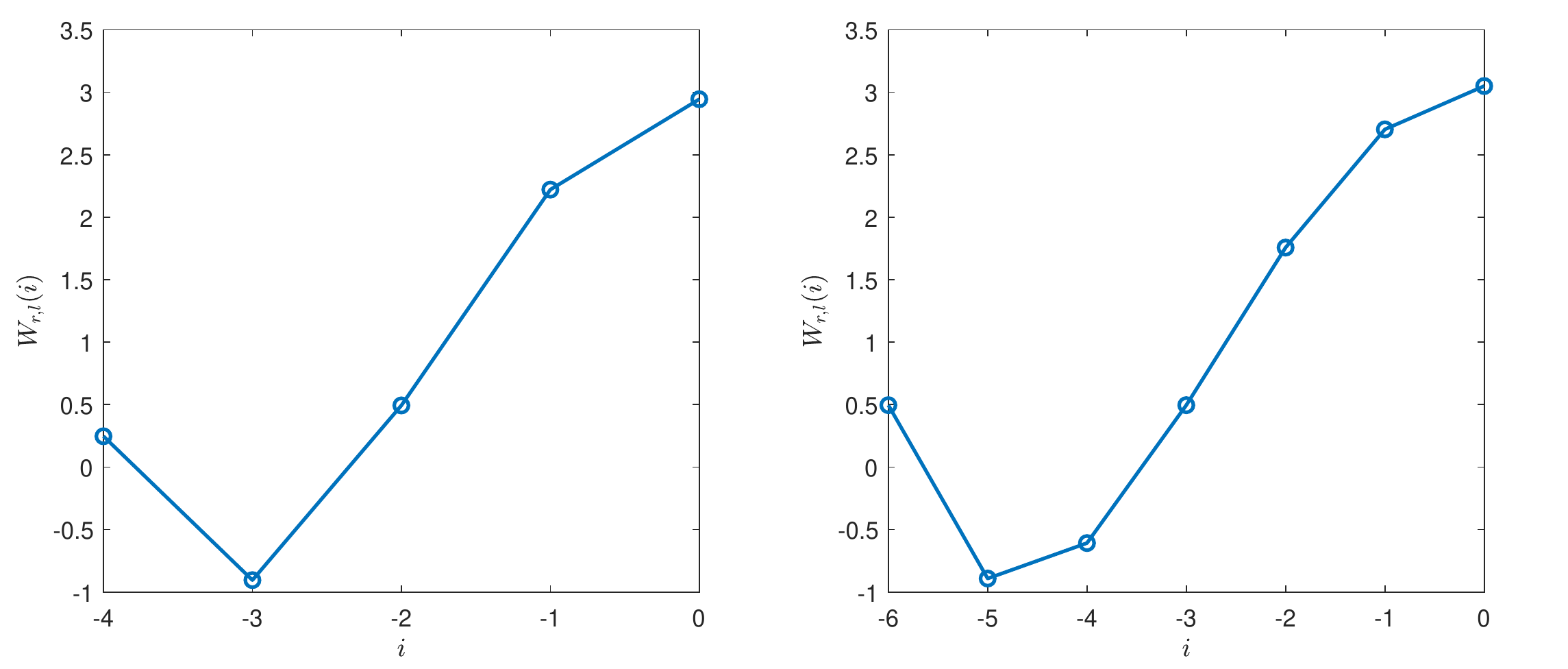}
    \caption{\small Examples of Discrete Kernel. The first figure plots the value of discrete kernel $W_{r,\ell}(i)$ with bandwidth $r = 4$, smoothing parameter $\ell = 4$, and $i \in \mathcal{F}_{r}$. The second figure plots $W_{r,\ell}(i)$ with bandwidth $r = 6$, smoothing parameter $\ell = 4$, and $i \in \mathcal{F}_{r}$.}
    \label{Fig-Kernel}
\end{figure}

By construction, the $W_{r,\ell}$ is a discretized version of the order $\ell$ optimal continuous boundary kernel that weighs only the historical observations \citep{Gasser1979, Altman1990}. As shown in Figure \ref{Fig-Kernel}, the kernel assigns more recent similarity matrices higher scores. To choose an optimal bandwidth $r$, \cite{Pensky2019} proposed an adaptive estimation procedure using \cite{LepskiMammenSpokoiny1997}. We employ the same method (see supplement appendix \ref{sec:S1}) and construct the estimator for the smoothed similarity matrices:
\begin{equation} \label{eq01}
	\widehat{\mathcal{S}}_{t,r} = \dfrac{1}{|\mathcal{F}_{r}|}\sum_{i \in \mathcal{F}_{r}} W_{r,\ell}(i) S_{t+i}.
\end{equation}

\begin{remark}
    Different from \cite{Pensky2019}, we only apply a boundary kernel on the historical similarity matrices for community detection. The changes are to avoid the ``look-ahead bias" \citep{Kan1995} in asset pricing inference in subsequent sections. 
\end{remark}

Once we obtain $\widehat{\mathcal{S}}_{t,r}$, we create a singular value decomposition of $\widehat{S}_{t,r} = \widehat{U}_t\widehat{\Sigma}_t\widehat{V}_t^\top$ for each $t = 1, 2, \cdots, T$. As \cite{Lei2015} discussed, the matrix $\widehat{U}_t$ may now have more than $K$ distinct rows due to the degree correction, whereas the rows of $\widehat{U}_t$ still only point to at most $K$ directions. Therefore, we apply the spherical clustering algorithm to find a community structure among the rows of the normalized matrix $\widehat{U}^+_t$ with $\widehat{U}_{t}^{+}(i,*) = \widehat{U}_{t}(i,*)/\Vert \widehat{U}_t(i,*)\Vert$. More specifically, we consider the following spherical $k$-medians spectral clustering:
\begin{equation}
	\big\Vert \widehat{Z}_{R,t}^{+}\widehat{Y}_t  - \widehat{U}_{t}^{+}\big\Vert_{F}^2 \leq (1+\varepsilon) \min_{\substack{{Z}^+_{R,t} \in \mathcal{M}_{N^{R}_+,K} \\ {Y}_t \in \mathbb{R}^{K \times K}}} \big\Vert {Z}_{R,t}^{+}{Y}_t  - \widehat{{U}}_{t}^{+}\big\Vert_{F}^2
\end{equation}
where $Y_t$ is some rotation matrix. In the last step, we extend $\widehat{{Z}}_{R,t}^{+}$ to obtain $\widehat{{Z}}_{R,t}$ by adding $N - N^{R}_+$ canonical unit row vectors at the end. $\widehat{{Z}}_{R,t}$ is the estimate of $Z_{R,t}$ from this method. Same techniques are applied to $\widehat{V}_t$ to obtain column communities. We summarize the detailed steps in Algorithm \ref{algmdrt}.

\begin{algorithm}[htp!]
	\SetKwInOut{Input}{Input}
	\SetKwInOut{Output}{Output}
	
	\caption{CASC in the Dynamic DCCBM}\label{algmdrt}
	\Input{Adjacency matrices ${A}_t$ for $t = 1, \cdots, T$; \\ Covariates matrix $X$; \\ Number of row communities $K_R$ and number of column communities $K_C$; \\ Approximation parameter $\varepsilon$.}
	\Output{Membership matrices of rows and columns $Z_{R,t}$ and $Z_{C,t}$ for $t = 1, \cdots, T$.}
	\smallskip
	Calculate regularized graph Laplacian $L_{\tau, t}$.
	
	Estimate $\mathcal{S}_t$ by $\widehat{\mathcal{S}}_{t,r} $ as in (\ref{eq01}).
	
	Compute the singular value decomposition of $\widehat{\mathcal{S}}_{t,r} = \widehat{U}_t\widehat{\Sigma}_t\widehat{V}_t^{\top}$ for $t = 1, \cdots, T$. 
	
	Extract the first $K$ columns of $U_t$ and $V_t$ that correspond to the $K$ largest singular values in $\Sigma_t$, where $K = \min\{K_R, K_C\}$. Denote the resulting matrices $\widehat{U}_t^{K} \in \mathbbm{R}^{N \times K}$ and $\widehat{V}_t^{K} \in \mathbbm{R}^{N \times K}$.
	
	Let $N_+^R$ be the number of nonzero rows of $\widehat{U}^K_t$; then, obtain $\widehat{U}^K_{t+} \in \mathbbm{R}^{N_+^R \times K}$ consisting of normalized nonzero rows of $\widehat{U}^K_{t+}$; that is, $\widehat{U}_{t+}^K(i,*) = \widehat{U}^K_{t}(i,*) / \big\Vert \widehat{U}^K_{t}(i,*) \big\Vert$ for $i$ such that $\big\Vert \widehat{U}^K_{t}(i,*) \big\Vert > 0$.
	
	Similarly, let $N_+^C$ be the number of nonzero rows of $\widehat{V}^K_t$; then, obtain $\widehat{V}^K_{t+} \in \mathbbm{R}^{N_+^C \times K}$ consisting of normalized nonzero rows of $\widehat{V}^K_{t+}$; that is, $\widehat{V}_{t+}^K(i,*) = \widehat{V}^K_{t}(i,*) / \big\Vert \widehat{V}^K_{t}(i,*) \big\Vert$ for $i$ such that $\big\Vert \widehat{V}^K_{t}(i,*) \big\Vert > 0$.
	
	Apply the ($1+\varepsilon$)-approximate $k$-medians algorithm to cluster the rows (columns) of $\widehat{\mathcal{S}}_{t,r}$ into $K_R$ ($K_C$) communities by treating each row of $\widehat{U}^{K}_{t+}$ ($\widehat{V}^{K}_{t+}$) as a point in $\mathbbm{R}^K$ to obtain $\widehat{Z}^{+}_{R,t}$ ($\widehat{Z}^{+}_{C,t}$). 
	
	Extend $\widehat{Z}_{R,t}^+$ and $\widehat{Z}_{C,t}^+$ to obtain $\widehat{Z}_{R,t}$ and $\widehat{Z}_{C,t}$, respectively, by arbitrarily adding $N-N^R_+$ and $N-N^C_+$ canonical unit row vectors at the end, such as $\widehat{Z}_{R,t}(i) = (1, 0, \cdots, 0)$ and $\widehat{Z}_{C,t}(i) = (1, 0, \cdots, 0)$ for $i$ such that $\big\Vert \widehat{U}_t(i,*)\big\Vert = 0$ and $\big\Vert \widehat{V}_t(i,*)\big\Vert = 0$.
	
	Output $\widehat{Z}_{R,t}$ and $\widehat{Z}_{C,t}$.
\end{algorithm}

\subsection{Classification consistency}

In the subsequent analysis, we illustrate that the dynamic CASC achieves classification consistency uniformly over time. We first make some assumptions on the graph that generates the dynamic network. The major assumption required here is \textit{assortativity}, which ensures that the nodes within the same community are more likely to share an edge than nodes in two different communities. 
\begin{assumption}{1} \label{Ass01}
	The time-varying network is composed of a series of assortative graphs that are generated under the stochastic blockmodel with covariates, where the block probability $B(k,k) \geq B(k,k')$, $\forall k' \neq k$.
\end{assumption}

Intuitively, the more frequent the community membership changes, the less stable the network will be. Consequently, it becomes harder to utilize the information from the historical and future network structures to detect the communities in the present network structure. In Assumption \ref{Ass02}, we restrict the maximum number of nodes that switch memberships ($s$) to some finite number. Based on this assumption, the proportion of nodes that switch their memberships shrinks to 0 as the size of the network grows to infinity. Additionally, we can easily bind the dynamic behavior of clustering matrices ($Z_{R,t+r} - Z_{R,t}$ or $Z_{C,t+r} - Z_{C,t}$) by noting that there are at most $rs$ nonzero rows in the differenced matrix.  
\begin{assumption}{2} \label{Ass02}
	At most, $s < \infty$ number of nodes can switch their memberships between any consecutive time instances.
\end{assumption}

\begin{assumption}{3} \label{Ass03}
	For $1 \leq k \leq k' \leq K$, let $\varsigma_t = t/T$ denote the time instance, there exists a function $f(\cdot; k, k')$ such that $B_t(k,k') = f(\varsigma_t; k, k')$ and $f(\cdot; k, k') \in \Sigma(\beta, L)$, where $\Sigma(\beta, L)$ is a H{\"o}lder class of functions $f(\cdot)$ on $[0,1]$ such that $f(\cdot)$ are $\ell$ times differentiable and 
	\begin{equation}
		|f^{(\ell)}(x) - f^{(\ell)}(x')| \leq L|x - x'|^{\beta - \ell}, \text{ for any $x, x' \in [0,1]$},
	\end{equation}
	with $\ell$ being the largest integer smaller than $\beta$.
\end{assumption}
Assumption \ref{Ass03} states that neither the connection probabilities nor the community memberships change drastically over the horizons. Finally, to guarantee the performance of our clustering method, we impose some conditions to regularize the behavior of the covariate matrix and the eigenvalues of the similarity matrices. 

\begin{assumption}{4} \label{Ass04}
	Define $\underline{\delta} = \inf_t \lbrace \min\{\min_i \mathcal{D}_{R,t}(i,i), \min_i\mathcal{D}_{C,t}(i,i)\} \rbrace$ as the minimum average degree. Then, for any $\epsilon \in (0,1)$, assume:
	\begin{enumerate}
	    \item[(i)] $\underline{\delta} > 3\ln(16NT/\epsilon)$;
	    \item[(ii)] $\alpha_{\max} = \sup_t \alpha_t \leq 
	    (NR)^{-1}a$, where $a = \sqrt{\frac{3\ln(16NT/\epsilon)}{\underline{\delta}}}$.
	\end{enumerate}
\end{assumption}

\begin{remark}
    Assumption \ref{Ass04} embodies several important conditions for classification consistency. Condition (i) indicates the graph should be at least semi-dense. Condition (ii) ensures the population eigen-gap to be greater than the maximum of the absolute difference between the sum of covariate variances within a block and the mean of the sums across all blocks, which means covariate similarities will not dominate the clustering results.
\end{remark}

Following \cite{Rohe2016}, we define the ``$R$-misclustered" and ``$C$-misclustered" vertices as 
\begin{equation} \label{eqdef}
	\mathbb{M}^{p}_{t} = \left\lbrace \text{$i$: $\big\Vert C_{i,t}^p - \mathcal{C}_{i,t}^p\mathcal{O}_{t}^{p} \big\Vert > \big\Vert C_{i,t}^p - \mathcal{C}_{j,t}^p\mathcal{O}_{t}^{p}\big\Vert$, for any $j \neq i$}\right\rbrace, \quad p \in \{R, C\},
\end{equation}
where $C_{i,t}^p$ and $\mathcal{C}_{i,t}^p$ for $p \in \{R, C\}$ are the community centroids of the $i$th node at time $t$ generated using the $k$-medians clustering on the left/right singular vectors and the population left/right singular vectors, respectively. 

\begin{theorem} \label{thm02}
	Let $\lambda_{1,t} \geq \lambda_{2,t} \geq \cdots \geq \lambda_{K,t} > 0$ be the $K = \min\{K_R, K_C\}$ largest singular values of $\mathcal{S}_t$. Let ${Z}_{R,t} \in \mathcal{M}_{N,K_R}$, ${Z}_{C,t} \in \mathcal{M}_{N,K_C}$, and $P_{\max} = \max\{\max_{i,t}(Z_{R,t}^{\top}Z_{R,t})_{ii}, \max_{i,t}(Z_{C,t}^{\top}Z_{C,t})_{ii}\}$ denote the size of the largest block over the horizons. Then, under Assumptions \ref{Ass01}-\ref{Ass04} and $K_R \leq K_C$, the misclustering rate satisfies
	\begin{align*} 
		\sup_t\dfrac{\left|\mathbb{M}^R_t\right|}{N} &\leq \dfrac{c_2(\varepsilon)KW_{\max}^2}{m_{r}^{2}N\lambda_{K,\max}^2} \left\lbrace (6+c_1)\dfrac{b}{\underline{\delta}^{1/2}} + \dfrac{2K_C}{\underline{\delta}}(\sqrt{2P_{\max}rs} + 2P_{\max}) + \dfrac{P_{\max}L}{\underline{\delta}\cdot \ell!}\left(\dfrac{r}{T}\right)^{\beta}\right\rbrace^2, \\
		\sup_t\dfrac{\left|\mathbb{M}^C_t\right|}{N} &\leq \dfrac{c_3(\varepsilon)KW_{\max}^2}{m_c^2N\gamma_c^2\lambda_{K,\max}^2}\left\lbrace (6+c_1)\dfrac{b}{\underline{\delta}^{1/2}} + \dfrac{2K_C}{\underline{\delta}}(\sqrt{2P_{\max}rs} + 2P_{\max}) + \dfrac{P_{\max}L}{\underline{\delta}\cdot \ell!}\left(\dfrac{r}{T}\right)^{\beta}\right\rbrace^2,
	\end{align*}
	with a probability of at least $1 - \epsilon$, where $\lambda_{K,\max} = \max_t\{\lambda_{K,t}\}$, $c_2(\varepsilon) = 2^6(2+\varepsilon)^2$, $c_3(\varepsilon) = 2^7(2+\varepsilon)^2$, $b = \sqrt{3\ln(16NT/\epsilon)}$. $m_{r}$, $m_{c}$, and $\gamma_{c}$ are defined by equation (\ref{B_Mr}), (\ref{B_Mc}), and (\ref{eq_gamma}) in the supplementary appendix, respectively. 
\end{theorem}

\begin{remark}
	The upper bounds on misclustering rates obtained in the Theorem \ref{thm02} are non-asymptotic. To justify the asymptotic behavior of the misclustering rates, following \cite{Rohe2016}, we first have $m_r = O(\sqrt{K/N})$ and $m_c = O(\sqrt{K/N})$. Then, following \cite{FanLiaoMincheva2013}, we have $\lambda_{K,\max}$ diverging at rate $O(N)$. Based on these two results, we can easily find the convergence rate of the component outside of the bracket in error bound being $O(1/N^2)$. Taking each component inside the bracket diverges slower than $O(N)$, we conclude that the error bound will converge to 0 as $N, T \rightarrow \infty$.
\end{remark}

\begin{remark}
    The uniform upper bound is analogous to the static upper bounds in \cite{Binkiewicz2017} and \cite{ZhangRohe2018}. \cite{Abbe2020} pointed out that these bounds are not tight enough to inform the information gain in combining the network and node attributes. \cite{Deshpande2018}, \cite{Abbe2020}, and \cite{LuSen2020} tried to derive the information threshold for exact recovery under static network settings. It is definitely worth another theoretical paper to derive the uniform information bound for dynamic networks, but it is beyond the scope of the current paper. For empirical analysis, we only need consistency to guarantee the performance of our method. In subsequent sections, we will illustrate the improvement of classification accuracy by combining two sets of information through simulations and empirical tests.
\end{remark}

\subsection{Monte Carlo simulations}

To intuitively illustrate the theoretical results intuitively, we carry out several simulation studies using our algorithm and existing clustering methods under different model setups. Our benchmark algorithms for the directed networks are the degree-corrected DI-SIM (DI-SIM-DC) by \cite{Rohe2016} and the covariate-assisted DI-SIM for static networks (CA-DI-SIM-Stc) by \cite{ZhangRohe2018}. Since both methods are intended for static networks, we apply them to the adjacency matrix period by period and report the average misclustering rate over time.

We set the block probability matrix $B_t$ as
\begin{equation} \label{eqsim}
	B_t = 
	\begin{bmatrix}
		0.60 & 0.30 & 0.20 & 0.10 \\
		0.30 & 0.50 & 0.20 & 0.10 \\
		0.20 & 0.20 & 0.40 & 0.10 \\
		0.10 & 0.10 & 0.10 & 0.30
	\end{bmatrix}\times \dfrac{T+2t}{2T}, \text{ with } 1 \leq t \leq T,
\end{equation}
and set the order of the polynomials for kernel construction at $L = 4$ for all simulations. Then, we generate the initial clustering matrices $Z_{R,1}$ and $Z_{C,1}$ by randomly choosing one entry in each row and assign it to 1. For $t = 2, \cdots, T$, we fix the last $N-s$ rows of $Z_{R,t-1}$ and $Z_{C,t-1}$ and re-assign 1s in the first $s$ rows of $Z_{R,t-1}$ and $Z_{C,t-1}$, respectively, to mimic the 
community membership changes. Lastly, we assume that the number of communities $K_R = K_C = 4$ for directed network is known throughout the simulations. The time-invariant node covariates are $R = \lfloor \ln(NT) \rfloor$ dimensional with values $X \sim U(0,10)$. We replicate all experiments 1000 times, and the misclustering rate we report is the temporal average of the misclustering rates; that is, $T^{-1}\sum_{t=1}^{T}|\mathbb{M}^R_t|/N$ and $T^{-1}\sum_{t=1}^{T}|\mathbb{M}^C_t|/N$. 

We first examine the clustering performance with a growing network size. The number of vertices in the network varies from 20 to 200 with step size 20. The time span is $T = 10$ and the number of membership changes is $s = 10$. We summarize the results in Figure \ref{Fig-SimMisRates} (a) and (b). Evidently, as the size of the network increases, the misclustering rates of the CA-DI-SIM-Dyn decrease sharply and dominate both CA-DI-SIM-Stc and DI-SIM-DC throughout the simulation. The CA-DI-SIM-Stc dominates DI-SIM-DC in classification accuracy, which is consistent with the argument that some similarity is sufficient for the covariates to assist in the discovery of the graph structure \citep{Binkiewicz2017}. The dominating performance of CA-DI-SIM-Dyn over CA-DI-SIM-Stc proves that historical linkage information is crucial for improving classification accuracy.

\begin{figure}[htp!]
	\centering
	\begin{subfigure}[tph!]{0.48\textwidth}
		\includegraphics[width = \linewidth]{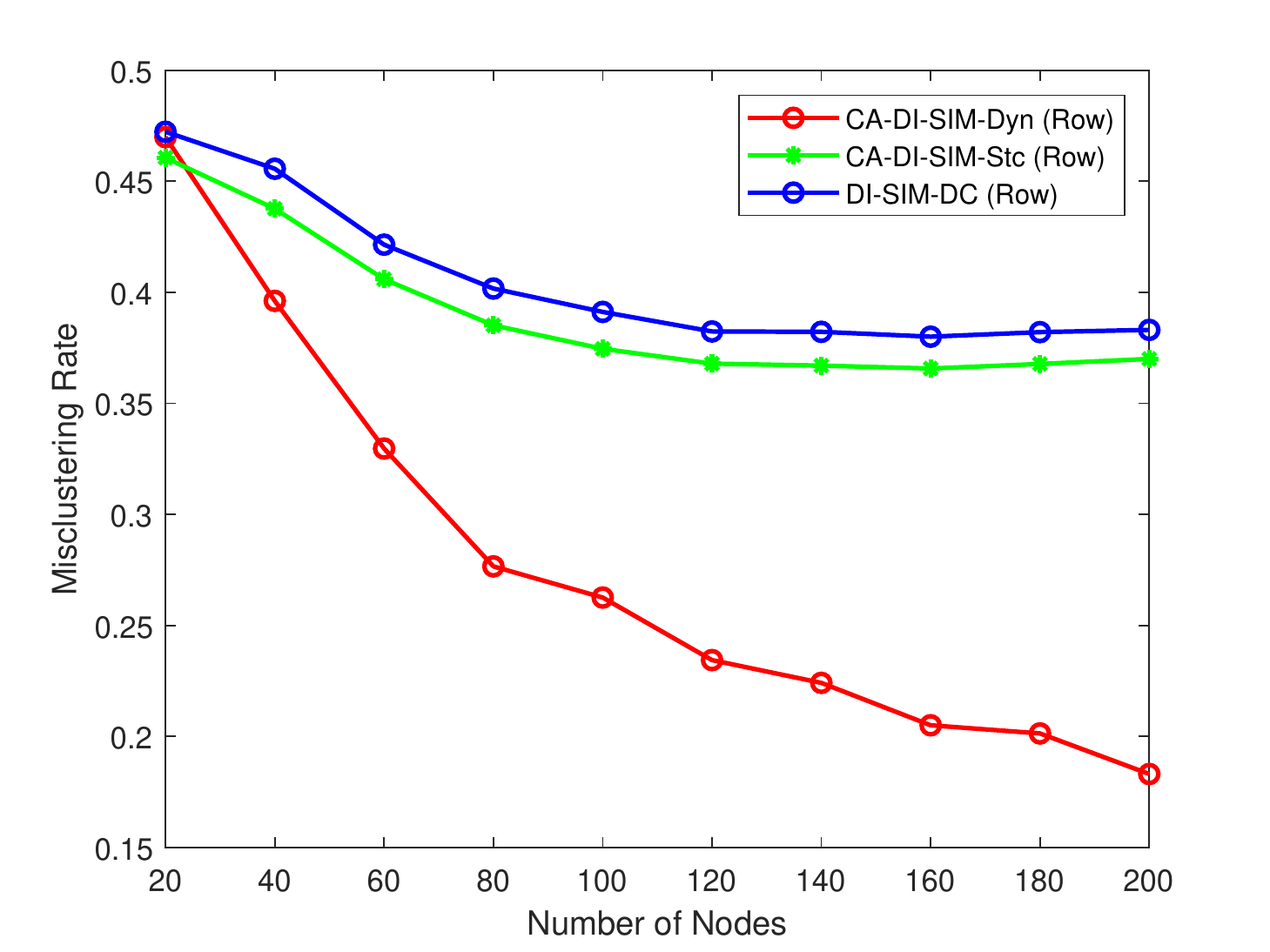}
		\caption{\footnotesize Directed Network (Row Community)}
	\end{subfigure}
	\begin{subfigure}[tph!]{0.48\textwidth}
		\includegraphics[width = \linewidth]{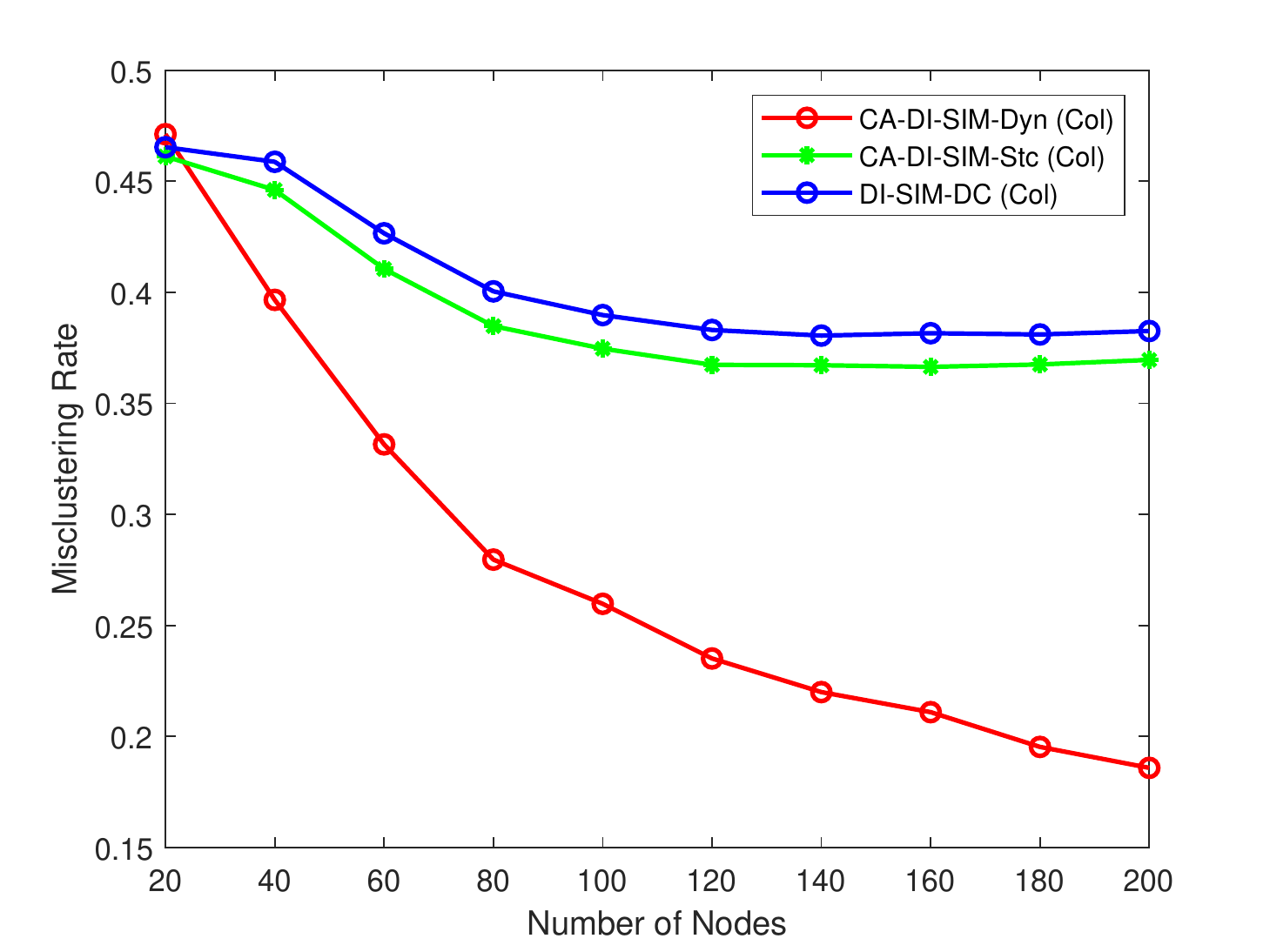}
		\caption{\footnotesize Directed Network (Column Community)}
	\end{subfigure} \\
	
	\begin{subfigure}[tph!]{0.48\textwidth}
		\includegraphics[width = \linewidth]{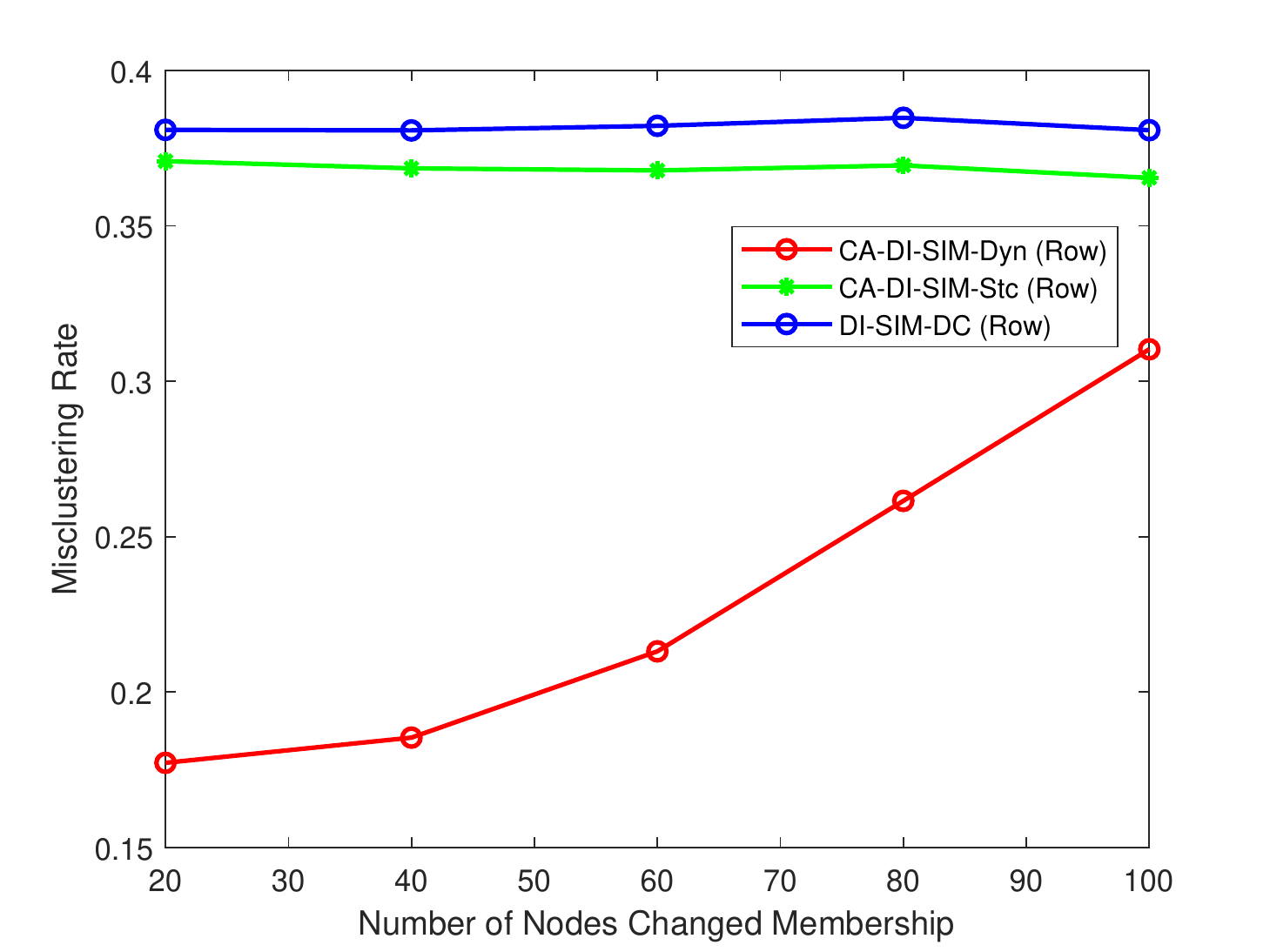}
		\caption{\small Directed Network (Row Community)}
	\end{subfigure}
	\begin{subfigure}[tph!]{0.48\textwidth}
		\includegraphics[width = \linewidth]{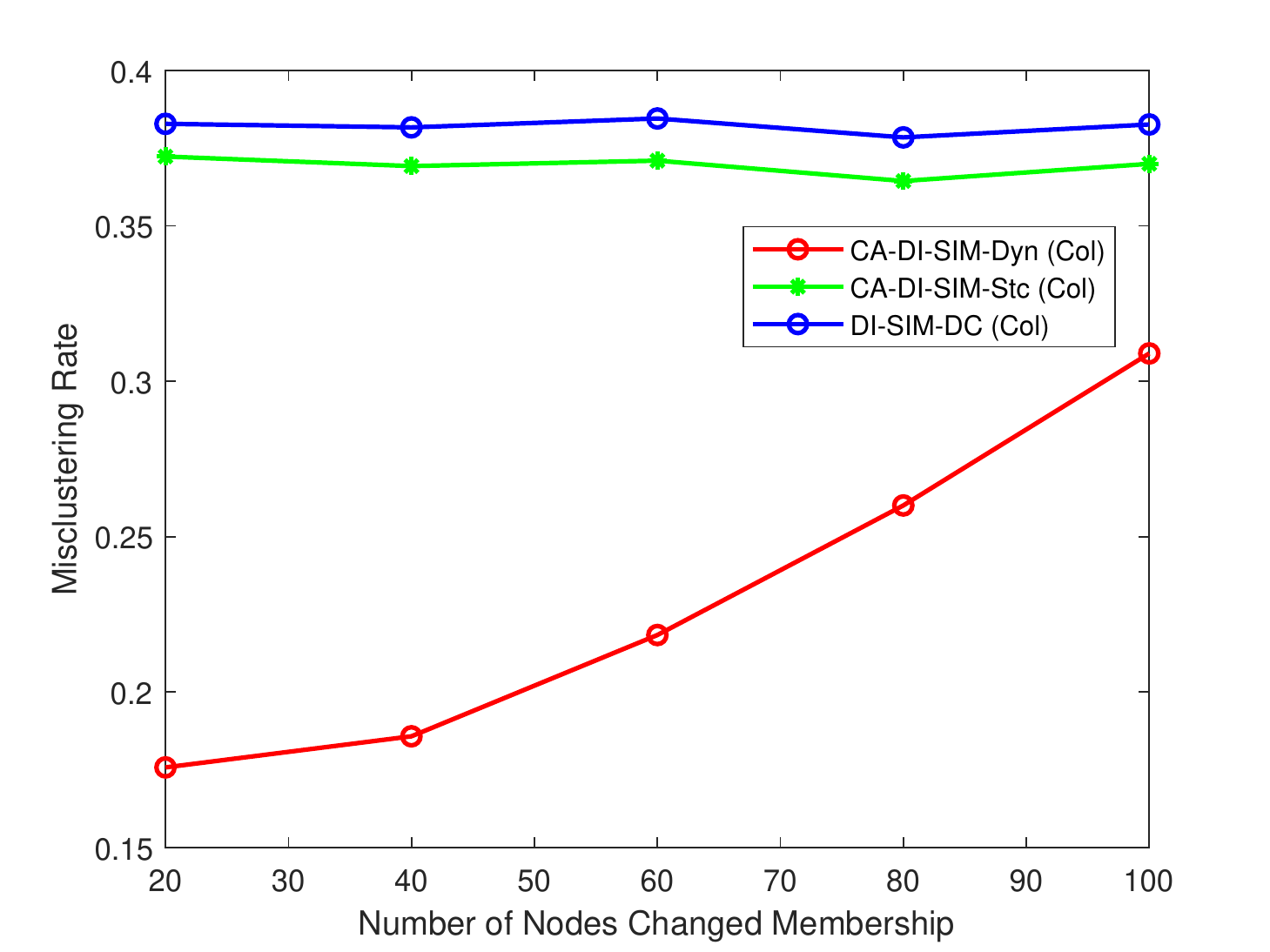}
		\caption{\small Directed Network (Column Community)}
	\end{subfigure}
	
	\caption{\small Misclustering Rates. This figure reports misclustering rates of different spectral clustering algorithms for networks with a growing number of vertices or with a growing number of membership changes. The CA-DI-SIM-Dyn represents Algorithm \ref{algmdrt}. The DI-SIM-DC is the degree-corrected DI-SIM in \cite{Rohe2016} and CA-DI-SIM-Stc is the static covariate-assisted DI-SIM method in \cite{ZhangRohe2018}. In subfigures (a) and (b), the number of nodes ($N$) grows from 20 to 200 with a step size of 20, and the number of membership changes is fixed at $s = 10$. In subfigures (c) and (d), the number of nodes is fixed at 200, and the number of membership changes ($s$) grows from 20 to 100 with stepsize of 20. The number of periods is $T = 10$, and the number of blocks is $K=4$, with the block probability matrix given in equation (\ref{eqsim}). The covariates are simulated with $R = \lfloor \ln(NT)\rfloor$ and $X \sim U(0,10)$. The smoothness parameter for kernel construction is $L = 4$. All simulations are conducted 1000 times.}
	\label{Fig-SimMisRates}
\end{figure}

Next, we check the relative performance for a growing maximum number of community membership changes. Here, we fix the total number of vertices at 200 and vary the community membership changes for each period, $s$, from 20 to 100 with step size 20. The total number of horizons is $T = 10$. We summarize the results in Figure \ref{Fig-SimMisRates} (c) and (d). Evidently, the misclustering rate of our method increases with the number of community membership changes. The benchmark methods are insensitive to the change of community memberships because they estimate communities using a one-period adjacency matrix. Despite the result, our method still achieves the lowest misclustering rate among all methods when the community memberships are relatively stable ($s<N/2$).

To perceptually illustrate the supreme accuracy of the dynamic CASC over the benchmark methods, we visualize the block structure of the dynamic network in Figure \ref{Fig-SimBlocks} by drawing the average adjacency matrices over a certain period and sort the rows and columns according to their community labels. Specifically, we simulate a dynamic network of 200 nodes with four blocks across 10 periods according to the block probability matrix $B_t$ in Equation (\ref{eqsim}) by setting parameter values $s = 10$, $R = \lfloor \ln(NT) \rfloor$, and $X \sim U(0,10)$. Then, we calculate the average adjacency matrix over six periods, $1/6\sum_{t=1}^{6}A_t$, and sort the rows and columns jointly according to their community labels. We choose six periods because the optimal bandwidth $r$ is chosen as 5 in the simulation. Thus, using six periods will lose no information in the kernel. Taking the average adjacency matrix as a $200\times200$ grid, we paint the square with a darker color when it has a larger value. 

\begin{figure}[htp!]
	\centering
	\includegraphics[width = \textwidth]{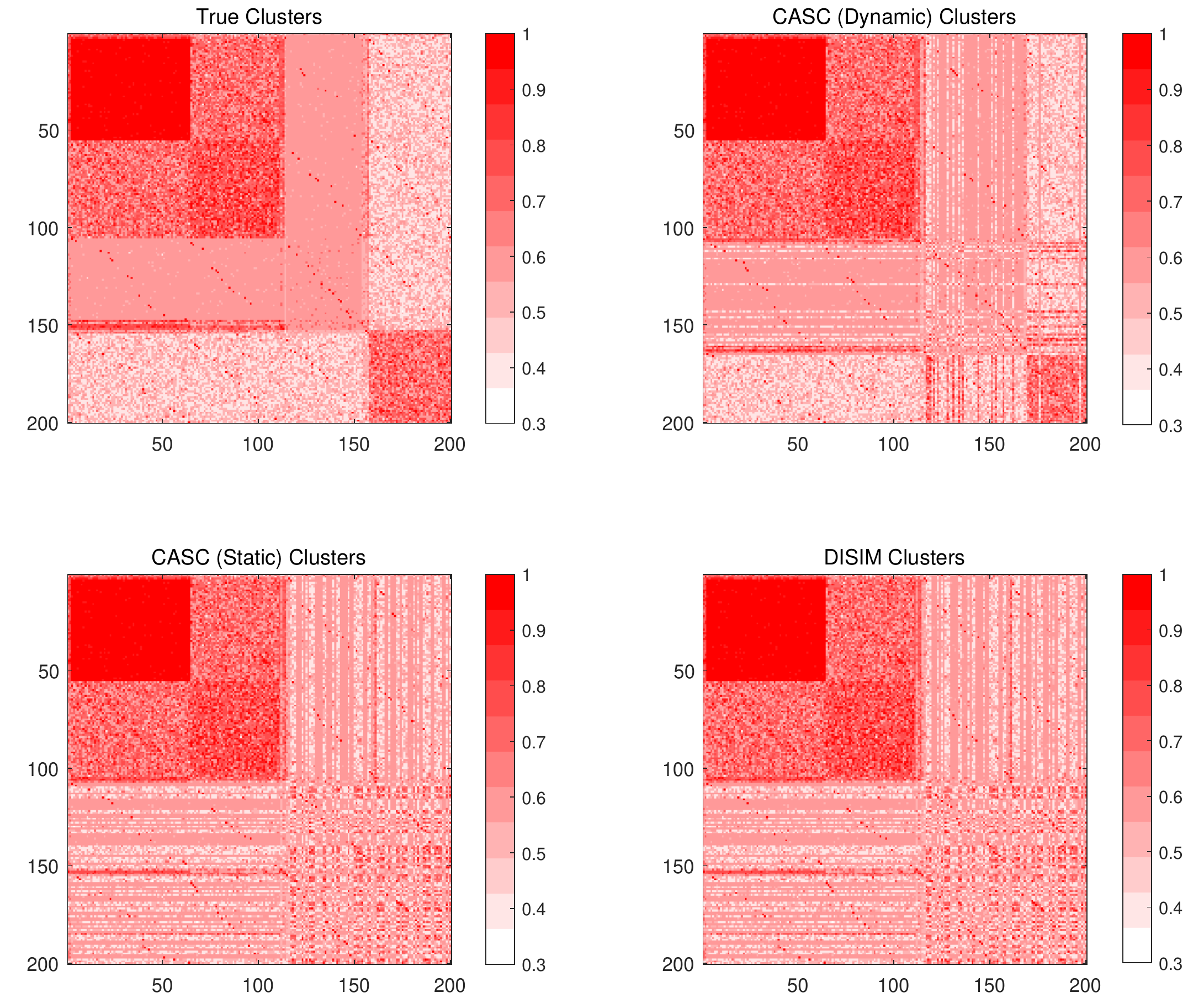}
	\caption{\small The True and the Estimated Blocks of the Network. This figure presents the block structure of the averaged adjacency matrices according to the true communities (simulation setting) and the estimated communities using different clustering methods. The averaged adjacency matrix reported is $1/6\sum_{t=1}^{6}A_{t}$. The adjacency matrices have $N = 200$ nodes with $K = 4$ communities. The number of periods is $T = 10$, and the block probability matrix is given in Equation (\ref{eqsim}). The covariates are simulated with $R = \lfloor \ln(NT)\rfloor$ and $X \sim U(0,10)$. The number of community member changes is fixed at $s = 20$, and the smoothness parameter for kernel construction is $L = 4$. All simulations are conducted 1000 times.}
	\label{Fig-SimBlocks}
\end{figure}

Evidently, the graph with true communities exhibits a clear block structure where each block is associated with reasonable connection frequencies abide by the block probability matrix. All three methods successfully recover the dense blocks (blocks in the northwest corner), whereas only the dynamic CASC can recover the less dense blocks (blocks in the southeast corner). The simulation results serve as remarkable evidence to confirm the advantages of the dynamic CASC in dealing with dense networks.

\subsection{Empirical communities} 

Applying the dynamic CASC algorithm to the cryptos network, we classify the crypto network into four column communities. We present three representative cryptos in each estimated community in Table \ref{Tab_Community}. It shows that as of December 31, 2017, the largest cryptos (BTC, ETH, and LTC) in terms of market capitalization are not necessarily categorized into the same community. For example, although the return patterns of BTC and ETH are closely related, their fundamental attributes are somewhat different: BTC employs SHA256, while ETH uses KECCAK-256. For comparison, we also report the representative cryptos of each community estimated by DI-SIM \citep{Rohe2016}. It shows that the representative cryptos of each community are pretty different when choosing different classification algorithms, indicating that the DI-SIM-estimated communities could have distinctive features compared to dynamic CASC-estimated communities.

\begin{table}[htp!]
	\caption{Representative Cryptos of Communities Estimated by Dynamic CASC and by DI-SIM.} \label{Tab_Community}
	
	\renewcommand*{\arraystretch}{1}
	{\small
		\begin{tabularx}{\textwidth}{@{\extracolsep{\fill}}lcccccc}
			\toprule
			& \multicolumn{3}{c}{Dynamic CASC} & \multicolumn{3}{c}{DI-SIM} \\ \cmidrule{2-4}\cmidrule{5-7}
			& 1 & 2 & 3 & 1 & 2 & 3 \\ 
			\midrule
			Community 1 & Bitcoin & Litecoin & Monero & Bitcoin & Bytecoin & Dogecoin \\
			& (BTC) & (LTC) & (XMR) & (BTC) & (BCN) & (DOGE) \\
			Community 2 & NEM   & Dash  & Syscoin & Dash  & BitcoinDark & BitBay \\
			& (XEM) & (DASH) & (SYS) & (DASH) & (BTCD) & (BAY) \\
			Community 3 & ReddCoin & Einsteinium & Crown & ReddCoin & Aeon  & OKCash \\
			& (RDD) & (EMC2) & (CRW) & (RDD) & (AEON) & (OK) \\
			Community 4 & Ripple & Ethereum & Augur & Ripple & Ethereum & Litecoin \\
			& (XRP) & (ETH) & (REP) & (XRP) & (ETH) & (LTC) \\
			\bottomrule
	\end{tabularx}
	}
	
	{\small NOTE: This table lists the top three cryptos in market capitalization of each community estimated by dynamic CASC and by DI-SIM on December 31, 2017. The crypto abbreviations are provided in the brackets below their names.}
\end{table}

To evaluate the classification efficiency under empirical settings, we define two important measures, the average within- and cross-community degrees:
\begin{align*}
	\textit{Within-Community Degrees} &= \dfrac{1}{T}\sum_{t = 1}^{T}\left[\dfrac{1}{N_C^2}\sum_{i \in \mathcal{C}}\sum_{j \in \mathcal{C}} A_{t}(i,j)\right], \\
	\textit{Cross-Community Degrees} &= \dfrac{1}{T}\sum_{t = 1}^{T}\left\lbrace \dfrac{1}{2N_C(N - N_C)}\sum_{i \in \mathcal{C}}\sum_{j \notin \mathcal{C}} \left[A_{t}(i,j) + A_{t}(j,i)\right]  \right\rbrace,
\end{align*}
where $N_C$ is the number of cryptos in community $\mathcal{C}$, and $N$ is the total number of cryptos. Intuitively, if the clustering method correctly classifies all cryptos, then the within-community degrees should be larger than the cross-community degrees, that is, the difference between the two measures should be positive. Table \ref{Tab_Degrees} summarizes the average within- and cross-community degrees of network linkages induced by the return cross-predictability, the technological similarities, and the joint similarities based on the dynamic CASC and the DI-SIM, respectively.

\begin{table}[htp!]
	\caption{Average Within- and Cross-Community Degrees} \label{Tab_Degrees}
	
	\renewcommand*{\arraystretch}{1}
	\setlength{\tabcolsep}{2pt}
	{\small
		\begin{tabularx}{\linewidth}{@{\extracolsep{\fill}}lS[table-format=1.2]S[table-format=1.2]S[table-format=1.2]S[table-format=1.2]S[table-format=2.1]S[table-format=2.1]S[table-format=1.2]S[table-format=2.2]S[table-format=2.1]S[table-format=2.1]S[table-format=1.2]S[table-format=1.2]}
			\toprule
			& \multicolumn{4}{c}{Return} & \multicolumn{4}{c}{Covariates} & \multicolumn{4}{c}{Combined} \\ \cmidrule{2-5}\cmidrule{6-9}\cmidrule{10-13}
			& {Within} & {Cross} & {Diff.} & {$t$-stat.} & {Within} & {Cross} & {Diff.} & {$t$-stat.} & {Within} & {Cross} & {Diff.} & {$t$-stat.} \\ \midrule
			\multicolumn{13}{l}{Panel A: Dynamic CASC} \\
			\quad Community 1 & 3.24  & 3.31  & -0.07 & -0.97 & 47.6  & 43.7  & 3.97  & 10.38 & 76.5  & 76.1  & 0.39  & 1.61 \\
			\quad Community 2 & 3.40  & 3.38  & 0.03  & 0.30  & 47.4  & 43.3  & 4.03  & 11.60 & 76.7  & 76.3  & 0.39  & 1.51 \\
			\quad Community 3 & 3.42  & 3.33  & 0.10  & 1.24  & 48.0  & 43.5  & 4.53  & 13.70 & 76.9  & 76.3  & 0.58  & 2.31 \\
			\quad Community 4 & 3.61  & 3.42  & 0.19  & 2.55  & 48.0  & 43.4  & 4.53  & 9.83  & 75.6  & 75.9  & -0.32 & -1.24 \\
			\midrule
			\multicolumn{13}{l}{Panel B: DI-SIM} \\
			\quad Community 1 & 2.58  & 3.04  & -0.46 & -3.74 & 44.4  & 44.2  & 0.20  & 0.01  & 72.0  & 76.1  & -4.09 & -6.26 \\
			\quad Community 2 & 2.38  & 2.83  & -0.45 & -6.16 & 52.4  & 47.2  & 5.20  & 14.30 & 88.9  & 81.0  & 7.96  & 9.63 \\
			\quad Community 3 & 2.27  & 2.62  & -0.35 & -2.50 & 38.3  & 41.2  & -2.00 & -3.09 & 74.1  & 75.5  & -1.46 & -3.84 \\
			\quad Community 4 & 5.82  & 4.21  & 1.61  & 6.36  & 53.5  & 45.9  & 7.60  & 41.24 & 78.7  & 77.4  & 1.30  & 5.73 \\
			\bottomrule
		\end{tabularx}
	}
	
	{\small NOTE: This table reports the average within- and cross-community degrees (in \%) for each community. Panel A reports the results based on the communities estimated by dynamic CASC and Panel B presents the results based on the communities estimated by DI-SIM. The first four columns report the results for the network linkages based on return cross-predictability. The middle four columns report results for the network linkages based on technological similarity. The last four columns present the results for the network linkages based on the combined information. The $t$-statistics are based on the Newey--West robust standard errors with 4 lags. The sample period ranges from January 1, 2016 to December 31, 2018.}
\end{table}

Consistent with our simulations, the dynamic CASC method shows superior classification efficiency than that of DI-SIM. It delivers more positive and significant differences in average degrees for all three types of network linkages. More interestingly, even though the DI-SIM only uses the return information to classify crypto, its performance for classifying return linkages is worse than that of the dynamic CASC, which needs to balance the return and technology information. This striking result proves that technological similarity assists extracting community information from return dynamics by emphasizing the return cross-predictability cause by common technological risk, and thus improves the classification.

\section{Economic Implications} \label{SecEcoMean}

In this section, we conduct three empirical tests to assess the economic contributions of crypto communities to the cryptos market. First, we examine whether the investors can diversify risks better by investing cryptos across different communities than those within the same community. Second, we explore the asset pricing implication of information propagation effects on within-community cryptos by constructing a cross-sectional portfolio that implements an inter-crypto momentum trading strategy. Finally, we dissect the portfolio return on several behavioral factors as falsification tests to support our informational mechanism.

\subsection{Risk diversification} 

Diversification is one of the major components of investment decision making under risk and uncertainty. Its underlying philosophy is to minimize both the probability of portfolio loss and its severity through multilateral insurance, in which each asset is insured by the remaining assets. Seminal papers by \cite{Markowitz1952} and \cite{Roy1952} mathematically formalized the idea of diversification of investment and showed that diversification can reduce risk without changing the expected portfolio returns. According to \cite{Rubinstein2002}, the key ingredient of risk diversification is the return covariance between risky assets that to be selected into the portfolio. The less positively correlated the assets, the lower the probability that the assets suffer losses simultaneously in the same proportion, and the better the protection offered by this multilateral insurance, namely diversification \citep{Koumou2020}. 

To testify that crypto communities facilitate risk diversification, we compare the out-of-sample average pairwise return correlations of cryptos within and across the communities. Specifically, we first estimate the crypto communities on each trading day and compute the pairwise Pearson correlation coefficients using the next 7- (1-week) and 30-day (1-month) crypto returns. Then, we calculate the average within-community correlations for each community by averaging the correlation coefficients of the crypto pairs whose members both belong to the focal community. We also obtain the average cross-community correlations for each community by averaging the correlation coefficients of the crypto pairs where only one of the members belongs to the focal community. Ideally, the average within-community correlations should be significantly higher than the average cross-community correlations for all communities because the cryptos within the same community share common sources of price and technological shocks. For comparison, we also present the results with the communities being estimated by DI-SIM.

\begin{table}[htp!]
	\caption{Within- and Cross-Community Average Return Correlations} \label{Tab_corr}
	
	\renewcommand*{\arraystretch}{1}
	{\small
		\begin{tabularx}{\linewidth}{@{\extracolsep{\fill}}lS[table-format=1.3]S[table-format=1.3]S[table-format=1.3]S[table-format=2.2]S[table-format=1.3]S[table-format=1.3]S[table-format=1.3]S[table-format=2.2]}
			\toprule
			& \multicolumn{4}{c}{7-Day Return Correlations} & \multicolumn{4}{c}{30-Day Return Correlations} \\ 
			\cmidrule{2-5}\cmidrule{6-9}
			& {Within} & {Cross} & {Diff.} & {$t$-stat.} & {Within} & {Cross} & {Diff.} & {$t$-stat.}\\ 
			\midrule
			\multicolumn{9}{l}{Panel A: Dynamic CASC} \\
			\quad Community 1 & 0.366 & 0.337 & 0.029 & 7.55  & 0.345 & 0.317 & 0.028 & 9.00 \\
			\quad Community 2 & 0.369 & 0.337 & 0.032 & 8.10  & 0.350 & 0.317 & 0.033 & 9.87 \\
			\quad Community 3 & 0.381 & 0.341 & 0.040 & 9.46  & 0.358 & 0.318 & 0.040 & 10.68 \\
			\quad Community 4 & 0.381 & 0.340 & 0.041 & 10.99 & 0.362 & 0.319 & 0.042 & 11.64 \\
			\midrule
			\multicolumn{9}{l}{Panel B: DI-SIM} \\
			\quad Community 1 & 0.336 & 0.328 & 0.008 & 1.24  & 0.322 & 0.309 & 0.014 & 0.85 \\
			\quad Community 2 & 0.295 & 0.307 & -0.012 & -1.02 & 0.268 & 0.284 & -0.016 & -0.43 \\
			\quad Community 3 & 0.562 & 0.388 & 0.173 & 10.15 & 0.550 & 0.369 & 0.181 & 4.34 \\
			\quad Community 4 & 0.296 & 0.310 & -0.014 & -1.78 & 0.266 & 0.284 & -0.018 & -1.48 \\
			\bottomrule
	    \end{tabularx}
	}
	
	{\small NOTE: This table reports the out-of-sample average return correlations of the cryptos within and across communities. Panel A presents the results for communities estimated by dynamic CASC and Panel B reports the results for communities estimated by DI-SIM. We first calculate the pairwise return correlations between cryptos in each period using the 7- and 30-day future returns. Then, we compute the within-community (cross-community) average return correlations using the correlation coefficients between the cryptos from the same community (different communities). The $t$-statistics reported are based on the Newey--West robust standard errors with 4 lags. The sample period spans from January 1, 2016 to December 31, 2018.}
	
\end{table}

Table \ref{Tab_corr} summarizes the results, and three findings are obtained. First, the return correlations between cryptos within the same community are, on average, significantly higher than those across communities. For example, in Panel A, the average within- and cross-community 30-day-ahead return correlations for Community 1 are 34.5 \% and 31.7\%, respectively. The correlation difference is 2.8\%, which is significant at the 1\% level with a Newey--West adjusted $t$-statistic of 9.0. Economically, the substantial reduction in return correlations suggests that investors can find attractive upside and diversification possibilities through allocating portfolio weights on cryptos from different communities. 

Second, the dynamic CASC-estimated communities exhibit better risk diversification performance than the DI-SIM-estimated communities. In particular, all the dynamic CASC-estimated communities possess significantly higher within-community return correlations than cross-community return correlations. By contrast, only one DI-SIM community (Community 3) shows significantly positive differences between the within- and cross-community return correlations. The consistent correlation patterns across the dynamic CASC-estimated communities make it more convenient for investors to form a diversified crypto portfolio following the cross-community principle. This is also empirical evidence that supports the validity of the dynamic CASC algorithm for community detection in time-varying networks.  

Finally, the differences between the within- and cross-community return correlations are robust to the length of the evaluation window. Specifically, the changes between the 7- and 30-day return correlation differences for dynamic CASC-estimated communities are around 0.1\%. From practitioners' perspective, a stable return correlation pattern allows the investors to make flexible investment decisions on diversifying short- or long-term risks. 

In summary, the dynamic CASC-estimated communities have meaningful implications for investors' portfolio choices and risk diversification through delivering lower cross-community return correlations robustly across all communities and over different horizons.  

\subsection{Momentum spillover} \label{SubsecMomentum}

Since information diffuses gradually across investors \citep{HongStein1999, HongLimStein2000}, the price of an asset may react sluggishly to the arrival of relevant news/information about other assets. This suggests that assets having fundamental similarities will have momentum spillovers, wherein past return of one asset predicts the returns of assets linked to it. Numerous papers have verified such spillovers in the equity market by using a variety of proxies for inter-firm linkages \citep[see, e.g.,][]{Moskowitz1999, Cohen2008, Menzly2010, Lee2019, Ali2020, Parsons2020}. Similarly, the estimated communities in the cryptos market are designed to optimally capture the information propagation and technological similarities between cryptos. Therefore, we expect the momentum spillover effects to exist within each community in the cryptos market. 

We construct a long-short portfolio that implements an inter-crypto momentum trading strategy to formally test the momentum spillover effects and justify the information propagation mechanism between similar cryptos in the network. The portfolio is constructed in three steps. First, for each crypto $i$ on day $t$, we compute the average return of all the other cryptos within the same community ($R_{it}$) to form a trading signal: 
\begin{align*}
    R_{it} = \dfrac{1}{|\mathcal{C}_{-i}|}\sum_{j \in \mathcal{C}_{-i}} r_{jt},
\end{align*}
where $\mathcal{C}_{-i}$ denotes the set of other cryptos that are in the same community as crypto $i$. Second, we sort cryptos into four quartiles according to the trading signal and label the cryptos in the top/bottom quartile as the winner/loser cryptos. Finally, we buy the winner cryptos and sell the loser cryptos with equal positions at the end of day $t$ and rebalance the trading positions at the end of next trading day.

We examine the portfolio return up to 7 days after the signal formation period. The hypothesis is that the information propagation mechanism should predict a drifting return pattern instead of a return reversal. Ideally, the winner/loser cryptos will generate higher/lower returns as the positive/negative shocks transmit through the network linkages between similar cryptos. 

\begin{table}[htp!]
	\caption{Average Returns of the Inter-Crypto Momentum Portfolio over the Horizons} \label{Tab_Portfolio}
	\renewcommand*{\arraystretch}{1}
	{\small
		\begin{tabularx}{\linewidth}{@{\extracolsep{\fill}}lS[table-format=1.2]S[table-format=1.2]S[table-format=1.2]S[table-format=1.2]S[table-format=1.2]S[table-format=1.2]S[table-format=1.2]}
			\toprule
			& $\textit{Ret}_{t+1}$ & $\textit{Ret}_{t+2}$ & $\textit{Ret}_{t+3}$ & $\textit{Ret}_{t+4}$ & $\textit{Ret}_{t+5}$ & $\textit{Ret}_{t+6}$ & $\textit{Ret}_{t+7}$ \\
			\midrule
			Loser & 0.06 & 0.55 & 0.55 & 0.51 & 0.69 & 0.58 & 0.61 \\
			2 & 0.38 & 0.53 & 0.58 & 0.53 & 0.58 & 0.49 & 0.51 \\
			3 & 0.75 & 0.63 & 0.50 & 0.60 & 0.45 & 0.49 & 0.63 \\
			Winner & 1.14 & 0.59 & 0.67 & 0.65 & 0.56 & 0.73 & 0.52 \\
			Winner$-$Loser & 1.08 & 0.05 & 0.12 & 0.14 & -0.13 & 0.16 & -0.09 \\
			$t$-statistic & 12.24 & 0.55 & 1.30 & 1.56 & -1.37 & 1.89 & -1.05 \\
			\bottomrule
		\end{tabularx}
	}
	
	{\small NOTE: This table reports the average future returns (in \%) of the inter-crypto momentum portfolio. The $t$-statistics are computed based on Newey--West robust standard errors with 4 lags. The sample period spans from January 1, 2016 to December 31, 2018.}
\end{table}

Table \ref{Tab_Portfolio} reports the results. Consistent with our hypothesis, the long-short portfolio generates a one-day-ahead average daily return of 1.08\% ($t$-stat. 12.24), statistically significant at the 1\% level. The results exhibit a strong monotonic relationship between past connected crypto returns and future crypto returns. More importantly, the positive portfolio return does not revert over the horizons. The portfolio analysis confirms that our network model and community structure effectively capture the information propagation between cryptos and provide economically significant investment opportunities. 

\subsection{Behavioral interpretation} \label{SubsecBehavioral}

The momentum spillover effect evidences strongly that return cross-predictability and technological similarities identify fundamental relationships between cryptos. However, the information propagation mechanism is not the only possibility that could generate the momentum spillover effect. Investors may be slow to carry information across cryptos due to a behavioral bias or constraint. Therefore, we explore three classical behavioral mechanisms commonly used to explain cross-sectional return anomalies (i.e., limit-to-arbitrage, investor attention, and economic policy uncertainty) to see if they can provide significant explanatory power for the momentum spillover effects in the cryptos market.

The first behavioral explanation is the limit-to-arbitrage. Cryptos market exhibits periods of large, recurrent arbitrage opportunities across exchanges \citep{Makarov2020}. According to \cite{ShleiferVishny1997}, sophisticated investors would quickly eliminate any return predictability arising from anomalies in a frictionless market without impediments to arbitrage. \cite{HouMoskowitz2005} documented that the asset prices will delay to respond to information due to the impact of market frictions. Therefore, when the cryptos market is more frictional, the information travels slower in the cryptos network, and delayed cryptos will not exhibit price momentum. Analogous to market friction measures in the equity market \citep{Brunnermeier2008, Brunnermeier2009}, we employ the cryptos market volatility index (VCRIX) constructed by \cite{Kim2019} and TED spread (TED) as the market-friction indicators. Specifically, the volatility index affects the market liquidity, and TED determines the investors' funding constraints. We split the sample into two subperiods (high and low) by cutting at the median of the market friction indicators. Then, we calculate and report the average portfolio returns within each subperiods in Table \ref{Tab_Frictions}.

\begin{table}[htp!]
	\caption{Market Frictions and the Inter-Crypto Momentum Portfolio Returns} \label{Tab_Frictions}
	\renewcommand*{\arraystretch}{1}
	{\small
		\begin{tabularx}{\linewidth}{@{\extracolsep{\fill}}lS[table-format=1.2]S[table-format=1.2]S[table-format=1.2]S[table-format=1.2]S[table-format=1.2]S[table-format=1.2]}
			\toprule
			& \multicolumn{3}{c}{VCRIX} & \multicolumn{3}{c}{TED} \\ 
			\cmidrule{2-4}\cmidrule{5-7}
			& {Low} & {High} & {High$-$Low} & {Low} & {High} & {High$-$Low} \\
			\midrule
			Loser   & -0.03 & 0.10 & 0.13 & 0.17 & -0.39 & -0.56 \\
			2     & 0.43 & 0.36 & -0.06 & 0.58 & -0.05 & -0.63 \\
			3     & 0.81 & 0.73 & -0.08 & 1.03 & -0.02 & -1.05 \\
			Winner  & 1.27 & 1.08 & -0.19 & 1.42 & 0.46 & -0.96 \\
			Winner$-$Loser & 1.30 & 0.98 & -0.32 & 1.24 & 0.85 & -0.39 \\
			$t$-statistic & 8.98 & 8.86 & -1.59 & 9.65 & 5.12 & -1.20 \\
			\bottomrule
		\end{tabularx}
	}
	
	{\small NOTE: This table reports the inter-crypto momentum portfolio returns (in \%) in subperiods of high and low market frictions. $t$-statistics are computed based on Newey--West robust standard errors with 4 lags. The sample period spans from January 1, 2016 to December 31, 2018.}
\end{table}

Evidently, the average portfolio returns do not exhibit significant differences across the high and low episodes of market frictions. The average portfolio returns are generally more positive, albeit insignificant, when market frictions are lower, which is consistent with \cite{HouMoskowitz2005}. In summary, the momentum spillover effect exists regardless of market frictions, indicating that the limit-to-arbitrage mechanism does not explain our crypto anomalies.  

The second behavioral explanation is investor attention. \cite{PengXiong2006} pointed out that investors have limited attention and are subject to overconfidence. Therefore, they need to process information to infer the value of cryptos and overestimate the precision of the acquired information due to overconfidence. \cite{GPTT2018} show that investor attention could spill over along the network linkages. Together, the limited attention, attention spillover, and overconfidence lead to an overreaction-induced momentum effect in the cryptos network. Following \cite{LiuTsyvinski2021}, we proxy the market-wide investor attention by using the daily Google search volume for the words ``crypto" and ``cryptocurrency," subtracting its 1-month historical average. Considering the leading position of Bitcoin in the cryptos market \citep{Griffin2020}, we also proxy the market-wide investor attention with the investor attention paid to Bitcoin, that is, the Google search volume for the word ``Bitcoin." Then, we repeat the test in Table \ref{Tab_Frictions} by splitting the sample into high and low attention subperiods. 

\begin{table}[htp!]
	\caption{Investor Attention and the Inter-Crypto Momentum Portfolio Returns} \label{Tab_Attention}
	\renewcommand*{\arraystretch}{1}
	{\small
		\begin{tabularx}{\linewidth}{@{\extracolsep{\fill}}lS[table-format=1.2]S[table-format=1.2]S[table-format=1.2]S[table-format=1.2]S[table-format=1.2]S[table-format=1.2]}
			\toprule
			& \multicolumn{3}{c}{Bitcoin Attention} & \multicolumn{3}{c}{Cryptos Attention} \\ \cmidrule{2-4}\cmidrule{5-7}
			& {Low} & {High} & {High$-$Low} & {Low} & {High} & {High$-$Low} \\
			\midrule
			Loser   & -0.02 & 0.20  & 0.22  & 0.33  & -0.44 & -0.77 \\
			2     & 0.24  & 0.66  & 0.42  & 0.71  & -0.20 & -0.90 \\
			3     & 0.48  & 1.26  & 0.78  & 0.95  & 0.41  & -0.53 \\
			Winner  & 1.02  & 1.36  & 0.33  & 1.43  & 0.63  & -0.80 \\
			Winner$-$Loser & 1.04  & 1.15  & 0.11  & 1.09  & 1.06  & -0.03 \\
			$t$-statistic & 9.61 & 7.54 & 0.55 & 9.58 & 7.70 & -0.16 \\
			\bottomrule
		\end{tabularx}
	}
	
	{\small NOTE: This table reports the inter-crypto momentum portfolio returns (in \%) in subperiods of high and low investor attention. $t$-statistics are computed based on Newey--West robust standard errors with 4 lags. The sample period spans from January 1, 2016 to December 31, 2018.}
\end{table}

The results are presented in Table \ref{Tab_Attention}. Clearly, the inter-crypto momentum portfolio returns do not exhibit significant differences in high and low episodes of investor attention. For example, the portfolio generates an average daily return of 1.15\% during low Bitcoin attention episodes while retaining a 1.04\% average daily return for the high ones. The difference is merely 0.11\% and statistically insignificant. The portfolio return patterns under high and low cryptos attention episodes are similar to that of Bitcoin attention. These results imply that investor attention is not the mechanism that drives the momentum spillover effects in the cryptos market.

Finally, observing that government policy and crypto prices have a strong synchronization \citep{Demir2018, ChengYen2020}, it is natural to wonder whether the portfolio returns are related to underlying policy uncertainty. By representing uncertainty with investor belief dispersion (disagreement), \cite{Atmaz2018} theoretically showed that higher disagreement leads to higher average bias and more overvaluation, thereby implying a negative disagreement--return relationship, as empirically documented in \cite{Brogaard2015}. Based on \cite{Daniel1998}, when there are misinformed investors who underreact to firm-specific news, the price momentum will occur. Therefore, more uncertainty leads to price momentum. To examine this hypothesis, we employ the U.S. and Global economic policy uncertainty indices (EPU and GEPU) introduced by \cite{Baker2016} to examine the portfolio performance in high and low uncertainty subperiods, respectively. 

\begin{table}[htp!]
	\caption{Economic Policy Uncertainty and the Inter-Crypto Momentum Portfolio Returns} \label{Tab_Uncertainty}
	\renewcommand*{\arraystretch}{1}
	{\small
		\begin{tabularx}{\linewidth}{@{\extracolsep{\fill}}lS[table-format=1.2]S[table-format=1.2]S[table-format=1.2]S[table-format=1.2]S[table-format=1.2]S[table-format=1.2]}
			\toprule
			& \multicolumn{3}{c}{EPU} & \multicolumn{3}{c}{GEPU} \\ \cmidrule{2-4}\cmidrule{5-7}
			& {Low} & {High} & {High$-$Low} & {Low} & {High} & {High$-$Low} \\
			\midrule
			Loser   & -0.30 & 0.38  & 0.68  & -2.01 & 0.16  & 2.17 \\
            2     & -0.10 & 0.83  & 0.93  & -1.76 & 0.49  & 2.25 \\
            3     & 0.51  & 0.98  & 0.46  & -1.61 & 0.87  & 2.49 \\
            Winner  & 0.71  & 1.53  & 0.83  & -1.07 & 1.25  & 2.32 \\
            Winner$-$Loser & 1.01  & 1.15  & 0.14  & 0.94  & 1.09  & 0.14 \\
            $t$-statistic & 7.70  & 8.88  & 0.77  & 3.07  & 11.91 & 0.33 \\
			\bottomrule
		\end{tabularx}
	}
	
	{\small NOTE: This table reports the inter-crypto momentum portfolio returns (in \%) in subperiods of high and low U.S. economic policy uncertainty (EPU) and global currency policy uncertainty (GEPU-Currency). $t$-statistics are computed based on Newey--West robust standard errors with 4 lags. The sample period spans from January 1, 2016 to December 31, 2018.}
\end{table}

Table \ref{Tab_Uncertainty} summarizes the results. It is evident that the portfolio returns are statistically indifferent across high and low uncertainty subperiods. Specifically, in high EPU episodes, the portfolio return is 1.15\% bps, whereas it is 1.01\% in low EPU episodes. The difference between the portfolio returns is 0.14\% with a Newey--West adjusted $t$-statistic of 0.77. When replacing the EPU with the GEPU indicator, the difference in portfolio returns remains the same magnitude and statistically insignificant. Both results indicate that policy uncertainty or investor disagreement has no impact on the momentum spillover effects in the cryptos market.

In summary, by dissecting the portfolio returns with limit-to-arbitrage, investor attention, and economic policy uncertainty factors, we falsify the potential behavioral mechanisms that would generate momentum spillover effects in the cryptos market. The joint evidence in the current section and Section \ref{SubsecMomentum} strongly supports that the momentum spillover effects result from information propagation between fundamentally similar cryptos.

\section{Conclusions} \label{SecConcl}

In this paper, we examine the market segmentation problem and uncover the information propagation mechanism in the cryptos market by constructing a time-varying network of cryptos that combines return cross-predictability and technological similarities. To analyze this network, we propose a dynamic CASC algorithm for community detection and prove its classification consistency. We demonstrate the performance of our methodology on simulated networks and show that it generally performs better than other state-of-the-art methods. Empirically, we apply dynamic CASC to the cryptos network and obtain four latent communities. Based on the estimated communities, we show that investors can better diversify risk through investing in cross-community cryptos and make a daily profit of 1.08\% through an inter-crypto momentum trading strategy within the communities. Finally, we dissect the portfolio returns on several behavioral mechanisms, namely limit-to-arbitrage, investor attention, and economic policy uncertainty, but find no explanatory power. These results confirm that information propagation rather than behavioral bias drives our findings.

This paper mainly focuses on the uniform consistency of the classification method. Recent literature has developed inspiring framework for statistical inference on membership profiles as well as on estimation of low-rank matrices \citep[e.g.,][]{Fan2019, ChenFan2019}. It would be an interesting and critical research direction to investigate how to accommodate dynamic networks and time-varying membership in these frameworks. We will leave them to future research.

\paragraph{Supplementary materials.}
The supplementary material contains details about tuning parameter choice, proofs of Theorem \ref{thm02} in Section \ref{SecMethod}.

\paragraph{Acknowledgements.}
The authors sincerely thank the editor, Professor Jianqing Fan, an associate editor, and two anonymous referees for their helpful and insightful comments that have significantly improved the article. We are also grateful to Ying Chen, Dashan Huang, Oliver Linton, Peter C.~B.~Phillips, Shuyang Sheng, Vladimir Spokoiny, Liangjun Su, Jun Yu, Yichong Zhang, and all participants of the 2017 ``Crypto-Currencies in a Digital Economy" workshop at Humboldt-Universit{\"a}t zu Berlin, the 2018 China Meeting of the Econometric Society, the 2019 SH3 Conference on Econometrics, the 2019 STAT of ML Conference, and the 2019 SoFiE 12th Annual Conference for their helpful discussions and comments. 

\paragraph{Funding.}
 Financial support of the European Union’s Horizon 2020 research and innovation program “FIN- TECH: A Financial supervision and Technology compliance training programme” under the grant agreement No 825215 (Topic: ICT-35-2018, Type of action: CSA), the European Cooperation in Science \& Technology COST Action grant CA19130 - Fintech and Artificial Intelligence in Finance - Towards a transparent financial industry, the Deutsche Forschungsgemeinschaft’s IRTG 1792 grant, the Yushan Scholar Program of Taiwan and the Czech Science Foundation’s grant no. 19-28231X / CAS: XDA 23020303 are greatly acknowledged.

\nocite{}
\bibliographystyle{chicago}
\bibliography{SBMBib}


\spacingset{1.8} 
\title{Supplementary Appendix to ``A Time-varying Network for Cryptocurrencies"}
	\author{Li Guo\thanks{Address: 600 Guoquan Rd, Shanghai 200433. Email: guo\_li@fudan.edu.cn. } \\
	{\small Fudan University}\\
	{\small Shanghai Institute of International Finance and Economics}\\
	\\
	Wolfgang Karl H{\"a}rdle\thanks{Address: Unter den Linden 6 10099 Berlin, Germany. Email: haerdle@hu-berlin.de. } \\
	{\small Humboldt-Universit{\"a}t zu Berlin, Singapore Management University} \\ 
	{\small Xiamen University, Charles University} \\
	\\
	Yubo Tao\thanks{Correspondence author. Address: 90 Stamford Rd, Singapore 178903. Email: ybtao@smu.edu.sg.} \\
	{\small Singapore Management University} \\
}
\maketitle

\newpage
\spacingset{1.4} 

\vspace*{2cm}
\begin{center}
	\LARGE{Supplementary Appendix to ``A Time-varying Network for Cryptocurrencies"}
\end{center}
\bigskip

This appendix comprises three sections. In section \ref{sec:S1}, we present how to choose the tuning parameters. In section \ref{sec:S2}, we provides the proofs of the main theorems in the above paper. In section \ref{sec:S3}, we present the necessary technical lemmas for proving the main results.

The notations that have been frequently used in the proofs are as follows: $[n] \stackrel{\operatorname{def}}{=} \{1,2,\cdots, n\}$ for any positive integer $n$, $\mathcal{M}_{m,n}$ be the set of all $m \times n$ matrices which have exactly one 1 and $n-1$ 0's in each row. $\mathbb{R}^{m \times n}$ denotes the set of all $m \times n$ real matrices. $\Vert \cdot \Vert$ is used to denote Euclidean $\ell_2$-norm for vectors in $\mathbb{R}^{m \times 1}$ and the spectral norm for matrices on $\mathbb{R}^{m \times n}$. $\Vert \cdot \Vert_\infty$ denotes the largest element of the matrix in absolute value. $\Vert \cdot \Vert_F$ is the Frobenius norm on $\mathbb{R}^{m \times n}$, namely $\Vert M \Vert_F \stackrel{\operatorname{def}}{=} \sqrt{\text{tr}(M^{\top}M)}$. $\Vert \cdot \Vert_{\phi_2}$ is the sub-Gaussian norm such that for any random variable $x$, there is $\Vert x \Vert_{\phi_2} \stackrel{\operatorname{def}}{=} \sup_{\kappa \geq 1} \kappa^{-1/2}(\Expt |x|^\kappa)^{1/\kappa}$. $\bm{1}_{m,n} \in \mathbb{R}^{m \times n}$ consists of all 1's, $\iota_n$ denotes the column vector with $n$ elements of all 1's. $\mathbf{1}_A$ denotes the indicator function of the event $A$. 

\section{Choice of Tuning Parameters} \label{sec:S1}

For the choice of $r$, we first need to determine the upper bound of the variance proportion of the estimation error $\Vert \widehat{\mathcal{S}}_{t,r} - \mathcal{S}_t\Vert$, which is $\Vert \widehat{\mathcal{S}}_{t,r} - \mathcal{S}_{t,r}\Vert$. In the following lemma, we derive a sharp probabilistic upper bound on $\Vert \widehat{\mathcal{S}}_{t,r} - \mathcal{S}_{t,r}\Vert$ using the device provided in \cite{Lei2015}.
\begin{lemma} \label{LemBS}
	Let $d = rN\Vert \mathcal{S}_t \Vert_{\infty}$ and $\eta \in (0,1)$. Then,
	\begin{equation*}
	\Vert \widehat{\mathcal{S}}_{t,r} - \mathcal{S}_{t,r}\Vert \leq (1-\eta)^{-2} \dfrac{W_{\max}\sqrt{d}}{r \vee 1},
	\end{equation*}
	with probability $1-\epsilon$, where $\epsilon = N^{\left(\frac{3}{16 \Vert \mathcal{S}_t \Vert_{\infty}} - 2\ln\left(\frac{7}{\eta}\right)\right)}$.
\end{lemma}

From Lemma \ref{LemBS} and the proofs of the previous theorems, we can see that $\Vert \widehat{\mathcal{S}}_{t,r} - \mathcal{S}_{t,r}\Vert$ is decreasing, while $\Vert \mathcal{S}_{t,r} - \mathcal{S}_{t}\Vert$ is increasing in $r$. Therefore, there exists an optimal $r^*$ that achieves the best bias-variance balance; that is,
\begin{equation}
r^* = \arg\min_{0 \leq r \leq T/2}\left((1-\eta)^{-2} \dfrac{W_{\max}\sqrt{d}}{r \vee 1} + \Vert \mathcal{S}_{t,r} - \mathcal{S}_t\Vert\right).
\end{equation}
Then, we can apply Lepski's method \citep{LepskiMammenSpokoiny1997} to construct the adaptive estimator for $r^*$. Without loss of generality, we choose $\eta = 1/2$. The, we define the adaptive estimator as
\begin{equation}
\widehat{r} = \max\left\lbrace 0 \leq r \leq T/2: \left\Vert \widehat{\mathcal{S}}_{t,r} - \widehat{\mathcal{S}}_{t, \rho} \right\Vert \leq 4W_{\max}\sqrt{\dfrac{N\Vert \mathcal{S}_t \Vert_{\infty}}{\rho \vee 1}}, \text{ for any $\rho < r$} \right\rbrace.
\end{equation}

Next, for the choice of $\alpha_{t}$, we select $\alpha_t$ to achieve a balance between $L_{\tau, t}$ and $C^{w}_t$:
\begin{equation}
\alpha_{t} = \dfrac{\lambda_{K}(L_{\tau, t}) - \lambda_{K+1}(L_{\tau, t})}{\lambda_{1}(C^{w}_t)}.
\end{equation}

Lastly, to determine $K$, we have several choices. \cite{Wang2017} implement a pseudo likelihood approach to choose the number of clusters in a stochastic blockmodel without covariates. \cite{ChenLei2017} propose a network cross-validation procedure to estimate the number of clusters by utilizing adjacency information. \cite{LiLevinaZhu2016} refine the network cross-validation approach by proposing an edge sampling algorithm. In our case, we apply the network cross-validation approach directly to the similarity matrix instead of the adjacency matrix. This is because the covariate matrix $C^{w}_t$ behaves just like an adjacency matrix when we use dummy variables to indicate different technology attributes. Therefore, the network cross-validation applies to the similarity matrix in our study.

\newpage
\newpage
\section{Proof of Main Results} \label{sec:S2}

\subsection{Proof of Theorem \ref{thm02}}
\begin{proof}
	In this proof, we deal with the clustering of left singular vector and the right singular vectors separately. 
	
	\bigskip
	\noindent \textit{(1) Clustering for $Z_{R,t}$.}
	First, by \cite{Rohe2016} and solution of ($1+\varepsilon$)-approximate $k$-means clustering, for each period $t = 1, \cdots, T$, we have
	\begin{equation} \label{B11}
		\dfrac{\left|\mathbb{M}^R_t\right|}{N} \leq \dfrac{8(2+\varepsilon)^2}{m_r^2N}\left\Vert U_t - \mathcal{U}_t\mathcal{O}_t\right\Vert_F^2,
	\end{equation}
	where 
	\begin{equation} \label{B_Mr}
	    m_r \stackrel{\operatorname{def}}{=} \min_{i,t}\{\min\{\Vert \Gamma_{R,t}(i,*) \Vert, \Vert \varGamma_{R,t}(i,*) \Vert\}\},
	\end{equation}
	and $\Gamma_{R,t}$ and $\varGamma_{R,t}$ are defined in Lemma \ref{lemA5}. 
	
	Then, by improved version of Davis-Kahn theorem from \cite{Lei2015}, we have
	\begin{equation} \label{B12}
		\left\Vert U_t - \mathcal{U}_t\mathcal{O}_t\right\Vert_F \leq \dfrac{2\sqrt{2K}}{\lambda_{K,t}}\left\Vert \text{sym}\left(\mathcal{S}_{t,r} - \mathcal{S}_{t}\right) \right\Vert,
	\end{equation}
	as $K_R \leq K_C$ and $K \stackrel{\operatorname{def}}{=} \min\{K_R, K_C\}$. Base on equation (\ref{B11}) and (\ref{B12}), we have
	\begin{equation} \label{B13}
		\sup_t\dfrac{\left|\mathbb{M}^R_t\right|}{N} \leq \dfrac{2^6(2+\varepsilon)^2K}{m_r^2N\lambda_{K,\max}^2} \sup_t \left\Vert \text{sym}\left(\mathcal{S}_{t,r} - \mathcal{S}_{t}\right) \right\Vert^2.
	\end{equation}
	
	Next, for $\mathcal{S}_t$, we have the following representation:
	\begin{equation}
		\mathcal{S}_t = \mathcal{D}_{R,t}^{-1/2}\varPsi^R Z_{R,t}B_{t}Z_{C,t}^{\top}\varPsi^C \mathcal{D}_{C,t}^{-1/2} + \alpha_t\mathcal{X}\mathcal{W}_t\mathcal{X}^{\top},
	\end{equation}
	where $\varPsi^p = \text{Diag}(\psi^p)$ with $p \in \{R, C\}$. Then, by definition of $\mathcal{S}_{t,r} \stackrel{\operatorname{def}}{=} |\mathcal{F}_r|^{-1}\sum_{i \in \mathcal{F}_r}W_{r,\ell}(i) \mathcal{S}_{t+i}$, we have the decomposition
	\begin{equation} \label{Bdcp}
		\Delta(r) = \sup_t\left\Vert \text{sym}\left(\widehat{\mathcal{S}}_{t,r} - \mathcal{S}_t\right) \right\Vert \leq \sup_t\left\Vert \text{sym}\left(\widehat{\mathcal{S}}_{t,r} - \mathcal{S}_{t,r}\right) \right\Vert + \sup_t\left\Vert \text{sym}\left(\mathcal{S}_{t,r} - \mathcal{S}_t\right) \right\Vert = \Delta_1(r) + \Delta_2(r).
	\end{equation}
	
	Now, we evaluate $\Delta_{1}(r)$ and $\Delta_{2}(r)$ respectively. For $\Delta_{1}(r)$, by Lemma \ref{lemA4}, we have
	\begin{align} \label{B14}
		\sup_t\left\Vert \text{sym}\left(\widehat{\mathcal{S}}_{t,r} - \mathcal{S}_{t,r}\right) \right\Vert &= \dfrac{1}{|\mathcal{F}_r|}\sum_{i \in \mathcal{F}_r}W_{r,\ell}(i) \sup_t\left\Vert \text{sym}\left(S_{t+i} - \mathcal{S}_{t+i}\right) \right\Vert \\ \nonumber
		&\leq W_{\max}(4+c_1)\left\lbrace \dfrac{3\log(16NT/\epsilon)}{\underline{\delta}}\right\rbrace^{1/2}.
	\end{align}
	
	For $\Delta_{2}(r)$, we first define
	\begin{equation}
		\widetilde{\mathcal{S}}_{t,r} = \dfrac{1}{|\mathcal{F}_{r}|}\sum_{i \in \mathcal{F}_{r}} W_{r,\ell}(i) \left(Y_{R,t}B_{t+i}Y_{C,t}^{\top} + \alpha_{t+i} \mathcal{X}\mathcal{W}_{t+i}\mathcal{X}\right).
	\end{equation}
	where $Y_{R,t} \stackrel{\operatorname{def}}{=} \mathcal{D}_{R,t}^{-1/2}\varPsi^R Z_{R,t}$ and $Y_{C,t} \stackrel{\operatorname{def}}{=} \mathcal{D}_{C,t}^{-1/2}\varPsi^C Z_{C,t}$. 
	
	Then, we decompose $\Delta_2(r)$ as 
	\begin{equation}
		\sup_t\left\Vert \text{sym}\left(\mathcal{S}_{t,r} - \mathcal{S}_t\right) \right\Vert \leq \sup_t\left\Vert \text{sym}\left(\mathcal{S}_{t,r} - \widetilde{\mathcal{S}}_{t,r}\right) \right\Vert + \sup_t\left\Vert \text{sym}\left(\widetilde{\mathcal{S}}_{t,r} - \mathcal{S}_t\right) \right\Vert = \Delta_{21}(r) + \Delta_{22}(r).
	\end{equation}
	For notation simplicity, we define $Y_{p,t} \stackrel{\operatorname{def}}{=} \mathcal{D}_{p,t}^{-1/2}\varPsi^{p} Z_{p,t}$ for $p \in \{R, C\}$, and the block diagonal matrix $\mathcal{Y}_t$ such that 
	\begin{equation}
		\mathcal{Y}_t \stackrel{\operatorname{def}}{=} \begin{bmatrix}
			Y_{R,t} & 0_{N \times K_C} \\
			0_{N \times K_R} & Y_{C,t}
		\end{bmatrix}.
	\end{equation}
	
	Note that
	\begin{align*}
		\left\Vert \text{sym}\left(\mathcal{S}_{t,r} - \widetilde{\mathcal{S}}_{t,r}\right) \right\Vert 	&\leq W_{\max}\max_{|i| \leq r}\left\Vert \text{sym}\left(Y_{R,t+i}B_{t+i}Y_{C,t+i}^{\top} - Y_{R,t}B_{t+i}Y_{C,t}^{\top}\right) \right\Vert \\
		&= W_{\max}\max_{|i| \leq r}\left\Vert \mathcal{Y}_{t+i}\text{sym}(B_{t+i})\mathcal{Y}_{t+i}^{\top} - \mathcal{Y}_{t}\text{sym}(B_{t+i})\mathcal{Y}_{t}^{\top}\right\Vert \\
		&\leq W_{\max}\max_{|i| \leq r}\left(\left\Vert \mathcal{Y}_{t+i} \right\Vert+ \left\Vert \mathcal{Y}_{t}\right\Vert\right) \left\Vert \text{sym}(B_{t+i})\right\Vert \left\Vert \mathcal{Y}_{t+i}-\mathcal{Y}_{t} \right\Vert,
	\end{align*}
	and $\left\Vert \text{sym}(B_{t+i})\right\Vert \leq K_C$ and $\Vert \varPsi^p\Vert \leq 1$ for $p \in \{R, C\}$, we then have
	\begin{align*}
		\left\Vert \mathcal{Y}_t \right\Vert &= \max\{\Vert\mathcal{D}_{R,t}^{-1/2}\varPsi^RZ_{R,t} \Vert, \Vert\mathcal{D}_{C,t}^{-1/2}\varPsi^CZ_{C,t} \Vert\} \\
		&\leq  \max\{\Vert\mathcal{D}_{R,t}^{-1/2}\Vert \Vert\varPsi^R\Vert \Vert Z_{R,t} \Vert, \Vert\mathcal{D}_{C,t}^{-1/2}\Vert \Vert\varPsi^C\Vert \Vert Z_{C,t} \Vert\} \\
		&\leq \underline{\delta'}^{-1/2}P_{\max}^{1/2}.
	\end{align*}
	
	Noticing that $\big\Vert \mathcal{D}_{p,t}^{-1/2} \big\Vert \leq \underline{\delta}^{-1/2}$ and $\Vert Z_{p,t}\Vert \leq P_{\max}^{1/2}$ for $p \in \{R, C\}$, and by Assumption \ref{Ass02}, $\Vert Z_{p,t+i} - Z_{p,t}\Vert \leq rs$, we then have
	\begin{align*}
		\left\Vert \mathcal{Y}_{t+i}-\mathcal{Y}_{t} \right\Vert &= \max\left\lbrace \left\Vert Y_{R,t+i} - Y_{R,t} \right\Vert, \left\Vert Y_{C,t+i} - Y_{C,t} \right\Vert\right\rbrace \\
		&\leq \max \left\lbrace \left\Vert \mathcal{D}_{R,t+i}^{-1/2}\varPsi^R Z_{R,t+i} - \mathcal{D}_{R,t}^{-1/2} \varPsi^R Z_{R,t} \right\Vert, \left\Vert \mathcal{D}_{C,t+i}^{-1/2}\varPsi^C Z_{C,t+i} - \mathcal{D}_{C,t}^{-1/2}\varPsi^C Z_{C,t} \right\Vert \right\rbrace \\
		&\leq \max \left\lbrace \left\Vert \mathcal{D}_{R,t+i}^{-1/2} \right\Vert \Vert\varPsi^R \Vert \Vert Z_{R,t+i} - Z_{R,t}\Vert + \left(\left\Vert \mathcal{D}_{R,t+i}^{-1/2} \right\Vert \Vert \varPsi^R \Vert + \left\Vert \mathcal{D}_{R,t}^{-1/2} \right\Vert \Vert\varPsi^R \Vert \right) \Vert Z_{R,t}\Vert \right., \\
		& \quad\quad\quad\quad   \left. \left\Vert \mathcal{D}_{C,t+i}^{-1/2} \right\Vert \Vert\varPsi^C \Vert \Vert Z_{C,t+i} - Z_{C,t}\Vert + \left( \left\Vert \mathcal{D}_{C,t+i}^{-1/2} \right\Vert \Vert\varPsi^C \Vert + \left\Vert \mathcal{D}_{C,t}^{-1/2} \right\Vert \Vert\varPsi^C \Vert \right) \Vert Z_{C,t}\Vert \right\rbrace \\
		&\leq \sqrt{\frac{2rs}{\underline{\delta}}} + \sqrt{\dfrac{4P_{\max}}{\underline{\delta}}}.
	\end{align*}
	
	Combining results above, we have
	\begin{align} \label{B18}
		\Delta_{21}(r) &= \sup_t\left\Vert \text{sym}\left(\mathcal{S}_{t,r} - \widetilde{\mathcal{S}}_{t,r}\right) \right\Vert \\ \nonumber
		&\leq W_{\max} \times (2{\underline{\delta}}^{-1/2}P_{\max}^{1/2}) \times K_C \times \left(\sqrt{\frac{2rs}{\underline{\delta}}} + \sqrt{\dfrac{4P_{\max}}{\underline{\delta}}}\right) \\ \nonumber
		&= \dfrac{2W_{\max}K_C}{\underline{\delta}}\left(\sqrt{2P_{\max}rs} + 2P_{\max}\right).
	\end{align}
	
	Lastly, for $\Delta_{22}(r)$, by Assumption \ref{Ass04} and the definition of $W_t$ giving $\Vert \text{sym}(W_t)\Vert \leq 1$, we have
	\begin{align*}
		\Delta_{22}(r) &= \sup_t\left\Vert \text{sym}\left(\widetilde{\mathcal{S}}_{t,r} - \mathcal{S}_{t}\right) \right\Vert \\
		&\leq \dfrac{1}{|\mathcal{F}_{r}|}\sum_{i \in \mathcal{F}_{r}} W_{r,\ell}(i) \sup_t\left(\left\Vert\text{sym}\left(Y_{R,t} (B_{t+i} - B_{t}) Y_{C,t}^{\top} \right)\right\Vert + \left\Vert \text{sym}\left(\alpha_{t+i}\mathcal{X}\mathcal{W}_{t+i}\mathcal{X} - \alpha_{t}\mathcal{X}\mathcal{W}_{t}\mathcal{X}\right) \right\Vert \right)\\
		&\leq \dfrac{1}{|\mathcal{F}_{r}|}\sum_{i \in \mathcal{F}_{r}} W_{r,\ell}(i)\left\Vert \mathcal{Y}_{t} \text{sym}(B_{t+i} - B_{t}) \mathcal{Y}_{t}^{\top} \right\Vert + 2\alpha_{\max}W_{\max}NR\sup_t \Vert \text{sym}(\mathcal{W}_t)\Vert \\
		&\leq \Delta_{22}^{(1)} + 2W_{\max}\left\lbrace \dfrac{3\log(16N/\epsilon)}{\underline{\delta}}\right\rbrace^{1/2}.
	\end{align*}
	
	Denote $Q_{r,t} = \dfrac{1}{|\mathcal{F}_{r}|}\sum_{i \in \mathcal{F}_{r}} W_{r,\ell}(i)(B_{t+i} - B_{t})$, and following \cite{Pensky2019}, we know for any $k, k' = 1, \cdots, K_C+K_R$,
	\begin{equation*}
	    |Q_{r,t}(k,k')| \leq \dfrac{LW_{\max}}{\ell!}\left(\dfrac{r}{T}\right)^{\beta}.
	\end{equation*}
	
	Then, for $\Delta_{22}^{(1)}$, we have
	\begin{align}
		\Delta_{22}^{(1)} &=  \dfrac{1}{|\mathcal{F}_{r,j}|}\sum_{i \in \mathcal{F}_{r,j}} W_{r,\ell}^{j}(i)\left\Vert \mathcal{Y}_{t} \text{sym}(B_{t+i} - B_{t}) \mathcal{Y}_{t}^{\top} \right\Vert \nonumber \\
		&\leq \max_{1\leq j' \leq 2N}\sum_{j = 1}^{2N}\left|(\mathcal{Y}_{t}Q_{r,t}\mathcal{Y}_{t}^{\top})(j,j')\right| \nonumber \\
		&\leq \max_{k,k'}|Q_{r,t}(k,k')|\max_{1\leq j' \leq 2N}\sum_{k=1}^{K_R+K_C}\sum_{k'=1}^{K_R+K_C}\left[\sum_{j \in \mathcal{G}_{t,k}} \mathcal{Y}_{t}(j,k)\right] \mathcal{Y}_{t}(j',k') \nonumber \\
		&\leq W_{\max}\dfrac{P_{\max}L}{\underline{\delta} \cdot \ell!}\left(\dfrac{r}{T}\right)^{\beta}, \label{B19}
	\end{align}
	where the last inequality comes from the fact that the spectral norm of a matrix is dominated by the $\ell_1$ norm.
	
	Therefore, combine equation (\ref{B13}), (\ref{Bdcp}), (\ref{B14}), (\ref{B18}), and (\ref{B19}), we obtain
	\begin{align*}
		\dfrac{\left|\mathbb{M}^R_t\right|}{N} \leq& \dfrac{c_2(\varepsilon)K_RW_{\max}^2}{m_r^2N\lambda_{K,\max}^2} \left\lbrace (6+c_1)\dfrac{b}{\underline{\delta}^{1/2}} + \dfrac{2K_C}{\underline{\delta}}(\sqrt{2P_{\max}rs} + 2P_{\max}) + \dfrac{P_{\max}L}{\underline{\delta}\cdot \ell!}\left(\dfrac{r}{T}\right)^{\beta}\right\rbrace^2,
	\end{align*}
	where $c_2(\varepsilon) = 2^6(2+\varepsilon)^2$, $b = \{3\log(16NT/\epsilon)\}^{1/2}$ and $\lambda_{K,\max} = \max_t \{\lambda_{K,t}\}$.
	
	\bigskip
	\noindent \textit{(2) Clustering for $Z_{C,t}$.}
	
	As shown in equation (\ref{Bdcp2}), the population regularized graph Laplacian of dynamic DCcBM has following decomposition
	\begin{equation} 
		\mathcal{S}_t = {\varPsi_{\tau,t}^R}^{1/2} Z_{R,t} \Omega_t Z_{C,t}^{\top}{\varPsi_{\tau,t}^C}^{1/2}
	\end{equation}
	Then, let $Y_{R, t} = Z_{R,t}^{\top}\varPsi_{\tau,t}^{R} Z_{R,t}$ and $Y_{C, t} = Z_{C,t}^{\top}\varPsi_{\tau,t}^{C} Z_{C,t}$, and
	\begin{equation}
		H_{\tau,t} = Y_{R,t}^{1/2}\Omega_t Y_{C,t}^{1/2}.
	\end{equation}
	Now, following \cite{Rohe2016}, we can define
	\begin{equation} \label{eq_gamma}
		\gamma_{c} \stackrel{\operatorname{def}}{=} \min_t\{\min_{i \neq j} \left\Vert H_{\tau,t}(*,i) - H_{\tau,t}(*,j)\right\Vert\},
	\end{equation}
	and thus
	\begin{equation} \label{B20}
		\dfrac{\left|\mathbb{M}^C_t\right|}{N} \leq \dfrac{16(2+\varepsilon)^2}{m_c^2N\gamma_{c}^2}\left\Vert U_t - \mathcal{U}_t\mathcal{O}_t\right\Vert_F^2.
	\end{equation}
	where
	\begin{equation} \label{B_Mc}
	    m_c \stackrel{\operatorname{def}}{=} \min_{i,t}\{\min\{\Vert \Gamma_{C,t}(i,*) \Vert, \Vert \varGamma_{C,t}(i,*) \Vert\}\}.
	\end{equation}
	
	Now, combining equation (\ref{B20}) with equations (\ref{B12}), (\ref{Bdcp}), (\ref{B14}), (\ref{B18}), and (\ref{B19}), we obtain
	\begin{align*}
		\sup_t \dfrac{\left|\mathbb{M}^C_t\right|}{N} \leq \dfrac{c_3(\varepsilon)K_RW_{\max}^2}{m_c^2N\gamma_c^2\lambda_{K_R,\max}^2} \left\lbrace (6+c_1)\dfrac{b}{\underline{\delta}^{1/2}} + \dfrac{2K_C}{\underline{\delta}}(\sqrt{2P_{\max}rs} + 2P_{\max}) + \dfrac{P_{\max}L}{\underline{\delta}\cdot \ell!}\left(\dfrac{r}{T}\right)^{\beta}\right\rbrace^2
	\end{align*}
	where $c_3(\varepsilon) = 2^7(2+\varepsilon)^2$, $b = \{3\log(16NT/\epsilon)\}^{1/2}$ and $\lambda_{K_R,\max} = \max_t \{\lambda_{K_R,t}\}$.
\end{proof}

\subsection{Proof of Lemma \ref{LemBS}}
\begin{proof}
	Firstly, by Lemma B.1 in Supplementary material of \cite{Lei2015}, fix $\eta \in (0,1)$, we have
	\begin{equation} \label{B21}
		\left\Vert \widehat{\mathcal{S}}_{t,r} - \mathcal{S}_{t,r}\right\Vert \leq (1-\eta)^{-2}\sup_{x,y \in \mathcal{T}}\left\vert x^{\top}(\widehat{\mathcal{S}}_{t,r} - \mathcal{S}_{t,r})y\right\vert,
	\end{equation}
	where $\mathcal{T} = \{x = (x_1, \cdots, x_N)\in \mathbbm{R}^N, \Vert x \Vert = 1, \sqrt{N}x_i/\eta \in \mathbbm{Z}, \forall i\}$. Then, let $d = rN\Vert \mathcal{S}_t \Vert_{\infty}$ with $r \geq 1$, we can split the pairs $(x_i, y_j)$ into light pairs 
	\begin{equation*}
		\mathscr{L} = \mathscr{L}(x,y) \stackrel{\operatorname{def}}{=} \{(i,j): |x_iy_j| \leq \sqrt{d}/N\},
	\end{equation*}
	and into heavy pairs
	\begin{equation*}
		\bar{\mathscr{L}} = \bar{\mathscr{L}}(x,y) \stackrel{\operatorname{def}}{=} \{(i,j): |x_iy_j| > \sqrt{d}/N\}.
	\end{equation*}
	
	For the light pair, first denote
	\begin{equation*}
		u_{ij} = x_iy_j\mathbf{1}_{\{|x_iy_j|\leq \sqrt{d}/N\}} + x_jy_i\mathbf{1}_{\{|x_jy_i|\leq \sqrt{d}/N\}},
	\end{equation*}
	then we have
	\begin{align*}
		\sum_{(i,j) \in \mathscr{L}(x,y)} &x_iy_j(\widehat{\mathcal{S}}_{t,r}(i,j) - \mathcal{S}_{t,r}(i,j)) \\
		&= \dfrac{1}{|\mathcal{F}_{r}|}\sum_{1 \leq i \leq j \leq N}\sum_{k \in \mathcal{F}_r} u_{ij}W_{r,\ell}(k)\left[S_{t+k}(i,j) - \mathcal{S}_{t+k}(i,j)\right].
	\end{align*}
	
	Denote $w_{ij} = |\mathcal{F}_{r}|^{-1}\sum_{k \in \mathcal{F}_r}W_{r,\ell}(k)\left[S_{t+k}(i,j) - \mathcal{S}_{t+k}(i,j)\right]$ and $\xi_{ij} = w_{ij}u_{ij}$, then we have $|w_{ij}| \leq W_{\max}\Vert \mathcal{S}_t \Vert_{\infty}$, and by \cite{Pensky2017}, it is known that $\xi_{ij}$ is a independent random variable with zero mean and absolute values bounded by $|\xi_{ij}| \leq 2W_{\max}\sqrt{r\Vert \mathcal{S}_t \Vert_{\infty}^3/N}$, using the fact that $|u_{ij}| \leq 2\sqrt{d}/N$.
	
	Now, applying Bernstein inequality, for any $c > 0$, we have
	\begin{align*}
		\Pr&\left(\sup_{x,y \in \mathcal{T}}\left\lvert \sum_{1 \leq i \leq j \leq N}\xi_{ij}\right\rvert \geq \dfrac{c\sqrt{d}}{r}\right) \\
		&\leq 2\exp\left(-\dfrac{\dfrac{c^2d}{2r}}{\sum_{1 \leq i \leq j \leq N}\mathbbm{E}\left(\xi_{ij}^2\right) + \dfrac{2W_{\max}}{3}\sqrt{\dfrac{r\Vert \mathcal{S}_t \Vert_{\infty}^3}{N}} \times \dfrac{c\sqrt{d}}{r}}\right) \\
		&\leq 2\exp\left(-\dfrac{\dfrac{c^2d}{2r}}{\left(\sum_{1 \leq i \leq j \leq N}u_{ij}^2\right)W_{\max}^2\Vert \mathcal{S}_t \Vert_{\infty}^2 + \dfrac{2W_{\max}}{3}\sqrt{\dfrac{r\Vert \mathcal{S}_t \Vert_{\infty}^3}{N}} \times \dfrac{c\sqrt{d}}{r}}\right) \\
		&\leq 2\exp\left(-\dfrac{3c^2N}{12W_{\max}^2\Vert \mathcal{S}_t \Vert_{\infty} + 4cW_{\max}\Vert \mathcal{S}_t \Vert_{\infty}}\right).
	\end{align*}
	
	Then, by a standard volume argument, we have the cardinality of $\mathcal{T} \leq \exp(N\log(7/\eta))$, and this ensures
	\begin{align} \label{BBnddrg}
		\Pr&\left(\sup_{x,y \in \mathcal{T}}\left\lvert \sum_{(i,j) \in \mathscr{L}(x,y)} x_iy_j(\widehat{\mathcal{S}}_{t,r}(i,j) - \mathcal{S}_{t,r}(i,j))\right\rvert \geq \dfrac{c\sqrt{d}}{r}\right) \nonumber\\ 
		&\leq \exp\left\lbrace - \left(\dfrac{3c^2}{12W_{\max}^2\Vert \mathcal{S}_t \Vert_{\infty} + 4cW_{\max}\Vert \mathcal{S}_t \Vert_{\infty}} - 2\log\left(\dfrac{7}{\eta}\right)\right)N\right\rbrace.
	\end{align}
	
	For the heavy pairs, we know
	\begin{align*}
		&\left\lvert \sum_{(i,j) \in \bar{\mathscr{L}}(x,y)} x_i y_j (\widehat{\mathcal{S}}_{t,r}(i,j) - \mathcal{S}_{t,r}(i,j))\right\rvert \\
		&= \left\lvert \dfrac{1}{|\mathcal{F}_r|}\sum_{(i,j) \in \bar{\mathscr{L}}(x,y)} x_i y_j \sum_{k \in \mathcal{F}_r} W_{r,\ell}(k) (S_{t+k}(i,j) - \mathcal{S}_{t+k}(i,j))\right\rvert  \\
		&\leq \left\lvert \dfrac{1}{|\mathcal{F}_r|}\sum_{(i,j) \in \bar{\mathscr{L}}(x,y)} \dfrac{x_i^2 y_j^2}{|x_iy_j|} \sum_{k \in \mathcal{F}_r} W_{r,\ell}(k) (S_{t+k}(i,j) - \mathcal{S}_{t+k}(i,j))\right\rvert \\
		&\leq \dfrac{N}{\sqrt{d}}W_{\max}\Vert \mathcal{S}_{t}\Vert_{\infty} \sum_{(i,j) \in \bar{\mathscr{L}}(x,y)} x_i^2 y_j^2 \\
		&= \dfrac{W_{\max}}{r}\sqrt{d} \sum_{(i,j) \in \bar{\mathscr{L}}(x,y)}x_i^2 y_j^2 \\
		&\leq \dfrac{W_{\max}}{r}\sqrt{d}.
	\end{align*}
	Therefore, choosing $c = W_{\max}$ in equation (\ref{BBnddrg}), we have
	\begin{equation} \label{B22}
		\Pr\left(\sup_{x,y \in \mathcal{T}}\left\lvert \sum_{1 \leq i \leq j \leq N} x_iy_j(\widehat{\mathcal{S}}_{t,r}(i,j) - \mathcal{S}_{t,r}(i,j))\right\rvert \leq \dfrac{W_{\max}\sqrt{d}}{r}\right) \geq 1 - \epsilon 
	\end{equation}
	where $\epsilon = N^{\left(\frac{3}{16 \Vert \mathcal{S}_t \Vert_{\infty}} - 2\log\left(\frac{7}{\eta}\right)\right)}$.
	
	In the end, by equation (\ref{B21}) and (\ref{B22}), we obtain, with probability $1-\epsilon$,
	\begin{equation*}
		\left\Vert \widehat{\mathcal{S}}_{t,r} - \mathcal{S}_{t,r}\right\Vert \leq (1-\eta)^{-2} \sup_{x,y \in \mathcal{T}}\left\vert x^{\top}(\widehat{\mathcal{S}}_{t,r} - \mathcal{S}_{t,r})y\right\vert \leq (1-\eta)^{-2} \dfrac{W_{\max}\sqrt{d}}{r}.
	\end{equation*}
\end{proof}

\newpage
\section{Technical Lemmas} \label{sec:S3}

\begin{lemma} \label{lemA4}
	Under Assumption \ref{Ass04}, for any $\epsilon > 0$, we have
	\begin{equation}
	\sup_t\left\Vert \text{sym}\left(S_{t} - \mathcal{S}_{t}\right) \right\Vert \leq \delta_{\max}\{3\log(16NT/\epsilon)\}^{1/2},
	\end{equation}
	with probability at least $1-\epsilon$
\end{lemma}
\begin{proof}
	By triangular inequality, we have
	\begin{align}
	\sup_t\left\Vert \text{sym}\left(S_{t} - \mathcal{S}_{t}\right)\right\Vert &\leq \sup_t\left\Vert \text{sym}\left(\alpha_tXW_tX^{\top} - \alpha_t\mathcal{X}\mathcal{W}_t\mathcal{X}^{\top}\right) \right\Vert \label{A04}\\
	&\ \ \ + \sup_t\left\Vert \text{sym}\left(\mathcal{D}_{R,t}^{-1/2}A_t\mathcal{D}_{C,t}^{-1/2} - \mathcal{D}_{R,t}^{-1/2}\mathcal{A}_t\mathcal{D}_{C,t}^{-1/2}\right) \right\Vert \label{A05} \\
	&\ \ \ + \sup_t\left\Vert \text{sym}\left(D_{R,t}^{-1/2}A_tD_{C,t}^{-1/2} - \mathcal{D}_{R,t}^{-1/2}A_t\mathcal{D}_{C,t}^{-1/2}\right) \right\Vert. \label{A06}
	\end{align}
	
	For equation (\ref{A04}), the spectral norm of the symmetrized $\alpha_tXW_tX^{\top} - \alpha_t\mathcal{X}\mathcal{W}_t\mathcal{X}^{\top}$ is bounded by
	\begin{align*}
	\sup_t&\left\Vert \text{sym}\left(\alpha_tXW_tX^{\top} - \alpha_t\mathcal{X}\mathcal{W}_t\mathcal{X}^{\top}\right) \right\Vert \\
	&= \alpha_{\max}\sup_t\left\Vert \text{sym}\left(X(W_t - \mathcal{W}_t)X^{\top}\right)\right\Vert + \alpha_{\max}\sup_t\left\Vert \text{sym}\left(X\mathcal{W}_tX^{\top}  - \mathcal{X}\mathcal{W}_t\mathcal{X}^{\top}\right) \right\Vert \\
	&\leq \alpha_{\max}NR \sup_t\Vert \text{sym}\left(W_t - \mathcal{W}_t\right) \Vert + 2\alpha_{\max}NR \sup_t\Vert \text{sym}\left(\mathcal{W}_t\right) \Vert \\
	&= \mathcal{O}_p(\alpha_{\max}NR).
	\end{align*}
	So, by Assumption \ref{Ass04}, we know that for large enough $N$, with probability at least $1-\epsilon/2$,
	\begin{equation*}
	\sup_t\left\Vert \text{sym}\left(\alpha_tXW_tX^{\top} - \alpha_t\mathcal{X}\mathcal{W}_t\mathcal{X}^{\top}\right) \right\Vert \leq c_1 a,
	\end{equation*}
	for some constant $c_1$. 
	
	For equation (\ref{A05}), by Assumption \ref{Ass04} and matrix Bernstein inequality, we know under assumption $\underline{\delta} > 3\log(16NT/\epsilon)$, $a < 1$. Therefore, we have
	\begin{align*}
	\Pr&\left(\sup_t \left\Vert \text{sym}\left(\mathcal{D}_{R,t}^{-1/2}A_t\mathcal{D}_{C,t}^{-1/2} - \mathcal{D}_{R,t}^{-1/2}\mathcal{A}_t\mathcal{D}_{C,t}^{-1/2}\right) \right\Vert > a\right) \\
	&\leq \sum_{t=1}^{T} \Pr\left(\left\Vert \text{sym}\left(\mathcal{D}_{R,t}^{-1/2}A_t\mathcal{D}_{C,t}^{-1/2} - \mathcal{D}_{R,t}^{-1/2}\mathcal{A}_t\mathcal{D}_{C,t}^{-1/2}\right) \right\Vert > a\right) \\
	&\leq 4NT\exp\left(-\dfrac{3\log(16NT/\epsilon)/\underline{\delta}}{2/\underline{\delta} + 2a/(3\underline{\delta})}\right) \\
	&\leq 4NT\exp\left(-\log(16NT/\epsilon)\right) \\
	&= \epsilon/4.
	\end{align*}
	
	Lastly, for equation (\ref{A06}), by \cite{Rohe2016}, we know with probability at least $1 - \epsilon/2$,
	\begin{align*}
	\sup_t&\left\Vert \text{sym}\left(D_{R,t}^{-1/2}A_tD_{C,t}^{-1/2} - \mathcal{D}_{R,t}^{-1/2}A_t\mathcal{D}_{C,t}^{-1/2}\right) \right\Vert \\
	&= \sup_t\left\Vert \text{sym}\left(L_{\tau,t} - \mathcal{D}_{\tau,t}^{-1/2}D_{\tau,t}^{1/2}L_{\tau,t}D_{\tau,t}^{1/2}\mathcal{D}_{\tau,t}^{-1/2} \right)\right\Vert \\
	&= \sup_t \left\Vert \text{sym}\left((I - \mathcal{D}_{\tau,t}^{-1/2}D_{\tau,t}^{1/2})L_{\tau,t}D_{\tau,t}^{1/2}\mathcal{D}_{\tau,t}^{-1/2} + L_{\tau,t}(I - D_{\tau,t}^{1/2}\mathcal{D}_{\tau,t}^{-1/2}) \right)\right\Vert \\
	&\leq \sup_t \left\Vert \text{sym}\left(\mathcal{D}_{\tau,t}^{-1/2}D_{\tau,t}^{1/2} - I\right) \right\Vert \sup_t\left\Vert \text{sym}\left(\mathcal{D}_{\tau,t}^{-1/2}D_{\tau,t}^{1/2}\right) \right\Vert + \sup_t \left\Vert \text{sym}\left(\mathcal{D}_{\tau,t}^{-1/2}D_{\tau,t}^{1/2} - I\right) \right\Vert \\
	&\leq a^2 + 2a
	\end{align*}
	
	Therefore, combine the results above, we obtain the upper bound for $\left\Vert \text{sym}\left(S_{t} - \mathcal{S}_{t}\right) \right\Vert$, i.e., with probability at least $1 - \epsilon$,
	\begin{equation}
	\sup_t\left\Vert \text{sym}\left(S_{t} - \mathcal{S}_{t}\right) \right\Vert \leq a^2 + 3a + c_1 a \leq (4 + c_1)a = (4+c_1)\left\lbrace \dfrac{3\log(16NT/\epsilon)}{\underline{\delta}}\right\rbrace^{1/2}.
	\end{equation}
\end{proof}

\begin{lemma} \label{lemA5}
	Under the dynamic DCCBM with $K_R$ row blocks and $K_C$ column blocks, define $\varGamma_{R,t} \in \mathbbm{R}^{N \times K_R}$ with columns containing the top $K_R$ left singular vectors of $\mathcal{S}_t$ and $\varGamma_{C,t} \in \mathbbm{R}^{N \times K_C}$ with columns containing the top $K_C$ right singular vectors of $\mathcal{S}_t$. Then, under Assumption \ref{Ass04}, there exist orthogonal matrices $U_{R,t}$ and $U_{C,t}$ depending on $\tau_t$ for each $t = 1, \cdots, T$, such that for any $i,j = 1, \cdots, N$,
	\begin{equation*}
	\varGamma_{p,t} = {\varPsi_{\tau,t}^{p}}^{1/2} Z_{p,t} (Z_{p,t}^{\top} {\varPsi_{\tau,t}^{p}}^{1/2} Z_{p,t})^{-1/2}U_{p,t} 
	\end{equation*}
	and
	\begin{equation*}
	\varGamma_{p,t}^*(i,*) = \varGamma_{p,t}^*(j,*) \Longleftrightarrow Z_{p,t}(i,*) = Z_{p,t}(j,*).
	\end{equation*}
	where $\varGamma_{p,t}^*(i,*) = \varGamma_{p,t}(i,*)/\Vert \varGamma_{p,t}(i,*) \Vert$ with $p \in \{R, C\}$.
\end{lemma}
\begin{proof}
	Define $D_{B,t}^{R}$ and $D_{B,t}^{C}$ are diagonal matrices with entries $D_{B,t}^{R}(i,i) = \sum_{j=1}^{K}B_t(i,j)$ and $D_{B,t}^{C}(i,i) = \sum_{j}^{K}B_{t}(j,i)$, and $\varPsi_{\tau, t}^{p} = \text{Diag}(\psi_{\tau,t}^{p})$ with	$\psi_{\tau,t}^{p}(i) = \psi_{i}^{p}\dfrac{\mathcal{D}_{p,t}(i,i)}{\mathcal{D}_{p,t}(i,i) + \tau_{p,t}}$ for $p \in \{R,C\}$. Then under dynamic DCCBM, we have the decomposition below, 
	\begin{equation*}
	\mathcal{L}_{\tau,t} = \mathcal{D}_{R,t}^{-1/2}\mathcal{A}_t\mathcal{D}_{C,t}^{-1/2} = {\varPsi_{\tau, t}^{R}}^{1/2}Z_{R,t}B_{L,t}Z_{C,t}^{\top}{\varPsi_{\tau, t}^{C}}^{1/2},
	\end{equation*}
	where $B_{L,t} = \left(D_{B,t}^{R}\right)^{-1/2}B_t\left(D_{B,t}^{C}\right)^{-1/2}.$
	
	Define $M_{R,t}$ and $M_{C,t}$ such that $\mathcal{X} = \Expt(X) = {\varPsi_{\tau,t}^R}^{1/2} Z_{R,t}M_{R,t} = {\varPsi_{\tau,t}^C}^{1/2} Z_{C,t}M_{C,t}$, and $\Omega_t = B_{L,t} + \alpha_t M_{R,t}\mathcal{W}_tM_{C,t}^{\top}$, then we know
	\begin{equation} \label{Bdcp2}
	\mathcal{S}_t = {\varPsi_{\tau,t}^R}^{1/2} Z_{R,t} \Omega_t Z_{C,t}^{\top}{\varPsi_{\tau,t}^C}^{1/2}.
	\end{equation}
	
	Now, denote $Y_{R, t} = Z_{R,t}^{\top}\varPsi_{\tau,t}^{R} Z_{R,t}$ and $Y_{C, t} = Z_{C,t}^{\top}\varPsi_{\tau,t}^{C} Z_{C,t}$, and let $H_{\tau,t} = Y_{R,t}^{1/2}\Omega_t Y_{C,t}^{1/2}$. Then, by singular value decomposition, we have $H_{\tau,t} = U_{R,t}\Lambda_tU_{C,t}^{\top}$. Define $\varGamma_{R,t} = {\varPsi_{\tau,t}^{R}}^{1/2} Z_{R,t} Y_{R,t}^{-1/2}U_{R,t}$ and $\varGamma_{C,t} = {\varPsi_{\tau,t}^{C}}^{1/2} Z_{C,t} Y_{C,t}^{-1/2}U_{C,t}$, then, for $p \in \{R, C\}$,
	\begin{align*}
	\varGamma_{p,t}^{\top}\varGamma_{p,t} &= U_{p,t}^{\top} Y_{p,t}^{-1/2} Z_{p,t}^{\top} {\varPsi_{\tau,t}^{p}}^{1/2}{\varPsi_{\tau,t}^{p}}^{1/2} Z_{p,t} Y_{p,t}^{-1/2}U_{p,t} \\
	&= U_{p,t}^{\top} Y_{p,t}^{-1/2} Y_{p,t} Y_{p,t}^{-1/2}U_{p,t} \\
	&= U_{p,t}^{\top}U_{p,t} = I,
	\end{align*}
	and we have
	\begin{align*}
	&\varGamma_{R,t}\Lambda_t\varGamma_{C,t} = {\varPsi_{\tau,t}^{R}}^{1/2} Z_{R,t} Y_{R,t}^{-1/2}U_{R,t}\Lambda_t U_{C,t}^{\top}Y_{C,t}^{-1/2}Z_{C,t}^{\top}{\varPsi_{\tau,t}^{C}}^{1/2}  \\
	&= {\varPsi_{\tau,t}^{R}}^{1/2} Z_{R,t} Y_{R,t}^{-1/2}H_{\tau,t}Y_{C,t}^{-1/2}Z_{C,t}^{\top}{\varPsi_{\tau,t}^{C}}^{1/2} \\
	&= {\varPsi_{\tau,t}^{R}}^{1/2} Z_{R,t} Y_{R,t}^{-1/2}\left(Y_{R,t}^{1/2}\Omega_t Y_{C,t}^{1/2}\right)Y_{C,t}^{-1/2}Z_{C,t}^{\top}{\varPsi_{\tau,t}^{C}}^{1/2} \\
	&= {\varPsi_{\tau,t}^R}^{1/2} Z_{R,t} \Omega_t Z_{C,t}^{\top}{\varPsi_{\tau,t}^C}^{1/2} = \mathcal{S}_t
	\end{align*}
	
	Following \cite{Rohe2016}, it is obvious that
	\begin{equation*}
	\varGamma_{p,t}^{*}(i,*) = \dfrac{\varGamma_{p,t}(i,*)}{\left\Vert \varGamma_{p,t}(i,*) \right\Vert} = Z_{p,t}(i,*)U_{p,t}, \text{ for $p \in \{R, C\}$},
	\end{equation*}
	which completes the proof.
\end{proof}

\nocite{}

\end{document}